\documentclass[12pt]{article}
	
\usepackage{amsmath, amsthm, amssymb}
\usepackage{graphicx}
\usepackage{enumerate}
\usepackage{natbib}
\usepackage{subfig}
\usepackage{multirow}
\usepackage{float}
\RequirePackage[colorlinks,citecolor=blue,urlcolor=blue]{hyperref}
\usepackage{mathtools}
\usepackage{fullpage}
\usepackage{indentfirst}
\usepackage{tikz}
\newtheorem{theorem}{Theorem}
\newtheorem{prop}{Proposition}
\newtheorem{lemma}{Lemma}
\newtheorem{cor}{Corollary}
\usepackage{xcolor}

\DeclareMathOperator*{\argmin}{argmin}

\DeclareMathOperator{\sign}{sign}

\DeclareMathOperator{\GCV}{GCV}

\DeclareMathOperator{\MSE}{MSE}
\DeclareMathOperator{\NS}{NS}

\definecolor{c1}{rgb}{0,  1, 0}
\definecolor{c2}{rgb}{0.2470588,  0.7490196, 0.2470588}
\definecolor{c3}{rgb}{0.4980392, 0.4980392, 0.4980392}
\definecolor{c4}{rgb}{0.7490196, 0.2470588, 0.7490196}
\definecolor{c5}{rgb}{1, 0, 1}

\DeclareRobustCommand\full  {\tikz[baseline=-0.6ex]\draw[c3, thick] (0,0)--(0.57,0);} 
\DeclareRobustCommand\dotted{\tikz[baseline=-0.6ex]\draw[c2,thick,dotted] (0,0)--(0.57,0);} 
\DeclareRobustCommand\denselydashed{\tikz[baseline=-0.6ex]\draw[c5, thick, dash pattern={on 7pt off 1.5pt}] (0,0)--(0.57,0);} 
\DeclareRobustCommand\dotdash {\tikz[baseline=-0.6ex]\draw[c4,thick,dash dot] (0,0)--(0.62,0);} 
\DeclareRobustCommand\dashed {\tikz[baseline=-0.6ex]\draw[c1,thick,dashed] (0,0)--(0.57,0);} 

\newtheorem{example}{Example}

\title{Robust optimal estimation of location from discretely sampled functional data}
\author{Ioannis Kalogridis and Stefan Van Aelst}
\date{%
    Department of Mathematics, KU Leuven, Belgium\\[2ex]%
    \today
}

\begin{document}
\maketitle

\begin{abstract}
Estimating location is a central problem in functional data analysis, yet most current estimation procedures either unrealistically assume completely observed trajectories or lack robustness with respect to the many kinds of anomalies one can encounter in the functional setting. To remedy these deficiencies we introduce the first class of optimal robust location estimators based on discretely sampled functional data. The proposed method is based on M-type smoothing spline estimation with repeated measurements and is suitable for both commonly and independently observed trajectories that are subject to measurement error. We show that under suitable assumptions the proposed family of estimators is minimax rate optimal both for commonly and independently observed trajectories and we illustrate its highly competitive performance and practical usefulness in a Monte-Carlo study and a real-data example involving recent Covid-19 data.
\end{abstract}


\section{Introduction}
\label{sec:ROS1}
In recent years, technological innovations and improved storage capabilities have led practitioners to observe and record increasingly complex high-dimensional data that are characterized by an underlying functional structure. Such data are nowadays commonly referred to as functional data and relevant research has been enjoying considerable popularity, following works such as \citep{Ramsay:1982}, \citep{Ramsay:1991, Rice:1991} and \citep{Ramsay:2005}. While the field of functional data analysis has become very broad with many specialized subpaths, see, e.g., \citep{Ferraty:2006, Horv:2012, Kokoszka:2017}, a central problem is inference regarding location functions $\mu$ of a random function $X$ with sample paths in some nicely-behaved function space.

In the early days of functional data analysis, it was commonly assumed that a sample of fully observed curves $X_1, \ldots, X_n$ was readily available to practitioners. This implies among others that the mean function can be optimally  estimated by the sample mean of the curves, see e.g., \citep{Horv:2012,Hsing:2015}. More recently, emphasis has been placed on the more realistic setting of discretely observed curves, possibly distorted by additive random noise. An early contribution is \citep{Yao:2005}, who proposed local linear estimation with repeated measurements and obtained uniform rates of convergence for sparsely observed trajectories. \citet{Deg:2008} derived consistency conditions for a broad family of weighted linear estimators and anticipated later developments by remarking that the rate of convergence with respect to the $\mathcal{L}^2$-norm can at most be of order $n^{-1}$ for densely observed data. The delicate interplay between sample size and sampling frequency was further illustrated by \citet{Li:2010}, who showed that the uniform rate of convergence is in between the parametric and non-parametric rates, the exact rate depending on the sampling frequency. A major breakthrough was obtained by \citet{Cai:2011}, who established minimax rates of convergence for mean function estimation and showed that these optimal rates can be achieved by a least-squares smoothing spline estimator with repeated measurements. \citet{Xiao:2020} demonstrated that these convergence rates can also be attained by a least-squares penalized spline estimator under similar assumptions.

A common drawback of all the estimation procedures above is the stringent conditions required for the errors because they rely on the minimization of $\mathcal{L}_2$ distances. Resistant estimation of central tendency in functional data is a problem that has received significantly less attention in the literature. An early contribution in the area was made by \citet{Wei:2006} who developed a quantile estimator for longitudinal data based on spline expansions. In the same vein, \citet{Lima:2019} proposed replacing the check loss with a more general loss function and similarly use an unpenalized spline expansion in order to recover the mean function from discretely sampled functional data. To the best of our knowledge, most other existing proposals require completely observed trajectories. \citet{CA:2006} proposed a functional equivalent of the trimmed mean estimator and showed its consistency in the ideal setting of fully observed trajectories. The robust spatial median has received the lion's share of attention in the literature. \citet{Gervini:2008} proved its root-n convergence rate for fully observed finite-dimensional trajectories. These convergence results were greatly strengthened by \citet{Cardot:2013} to cover separable Hilbert spaces of functions. An even deeper treatment of the spatial median was offered by \citet{Ch:2014}, who established general Glivenko-Cantelli and Donsker-type results for the empirical spatial distribution process in general Banach spaces. More recently, \citet{Sinova:2018} proposed a broad family of M-estimators for functional location and derived its consistency, again under the finite-dimensional assumption of the complete trajectories. 

Practitioners who seek robust estimates often overcome the discreteness of the data and the presence of measurement error, either by a pre-smoothing step or by altogether ignoring these problematic aspects of the data and directly applying any of the aforementioned robust estimators. Although these are popular strategies due to the lack of a better alternative, it should be stressed that their theoretical side-effects are not well-understood and consequently it is impossible to decide on best practices. A question that arises naturally is whether it is possible to construct theoretically-optimal robust estimators for functional location that can operate in the ubiquitous functional data setting of  curves that are discretely observed with measurement error. We answer this question positively by studying a broad class of estimators, including the resistant quantile and Huber-type smoothing splines among others. Our results do not require uncorrelated measurement errors nor existence of any moments of the error distributions for well-chosen robust loss functions, thereby providing an important relaxation of the conditions of \citep{Cai:2011}. The estimators can be efficiently computed with the convenient B-spline representation and well-established fast iterative algorithms, so that the associated computational burden is minimal, even if the trajectories are densely observed.

Our contributions are threefold. We propose a broad class of M-type smoothing spline estimators to estimate location functions based on discretely sampled functional data. Next to estimators for the central tendency, this class of estimators also includes quantile and expectile estimators. We study the rate of convergence of these estimators under weak assumptions for both common and independent designs and show that these estimators are rate-optimal. We do not only provide convergence rates for the location functionals, but also for their derivatives, which are often helpful in uncovering intricate features of functional data.

\section{The family of M-type smoothing spline estimators}
\label{sec:ROSM2}

Let us first assume that $\{X(t): t \in [0,1]\}$ is a process  with a well-defined mean function, that is, $\sup_{t \in [0,1]} \mathbb{E}\{ |X(t)| \} < \infty$. All subsequent arguments can be modified to encompass any location function, e.g., the median or expectile functions provided that these functions are also well-defined. Let $X_1, \ldots, X_n$ denote $n$ independent and identically distributed copies of $X$. Our goal is to recover the mean function $\mu(t) = E\{X(t)\}$ from noisy observations of the discretized curves:
\begin{equation}
\label{eq:ROSM1}
Y_{ij} = X_{i}(T_{ij}) + \zeta_{ij}, \quad (j = 1, \ldots, m_i;\ i=1, \ldots, n),
\end{equation}
where $T_{ij}$ are sampling points and $\zeta_{ij}$ are random noise variables. In \citep{Yao:2005, Cai:2011} and \citep{Xiao:2020} it is assumed that the noise variables $\zeta_{ij}$ are independent of the $X_i$ and independent and identically distributed with zero mean and finite variance. However, we allow for correlated errors that are not necessarily independent of the curves.  

We may view the errors $\zeta_{ij}$ as part of $n$ independent copies of the process $\zeta: (\Omega, \mathcal{A}, \mathbb{P}) \times [0,1] \to \mathbb{R}$ evaluated at the $T_{ij}$, that is, $\zeta_{ij} = \zeta_i(T_{ij}), j = 1, \ldots, m_i;\ i=1, \ldots, n$. Expressing model \eqref{eq:ROSM1} in terms of mean-deviations, we may equivalently write
\begin{equation}
\label{eq:ROSM2}
Y_{ij} = \mu(T_{ij}) + \epsilon_{ij}, \quad (j = 1, \ldots, m_i;\ i=1, \ldots, n),
\end{equation}
where $\epsilon_{ij} = \epsilon_i(T_{ij}) = X_i(T_{ij}) - \mu(T_{ij}) + \zeta_i(T_{ij})$ denotes the error process associated with the $i$th response evaluated at $T_{ij}$. The problem is thus reformulated as a regression problem with repeated measurements and possibly correlated errors. 

While there are many parametric and non-parametric estimators that can potentially be applied to the problem at hand, smoothing spline estimation arises naturally after we restrict $\mu$ to the Hilbert-Sobolev space of order $r$ denoted by $\mathcal{W}^{r,2}([0,1])$, viz,
\begin{align*}
\mathcal{W}^{r, 2}([0,1]) = \{ f:[0,1] \to \mathbb{R}, f\ &\text{has $r-1$ absolutely continuous derivatives}  \\  & f^{(1)}, \ldots, f^{(r-1)}\ \text{and} \int_0^1 |f^{(r)}(x)|^2 dx< \infty  \}.
\end{align*}
This separable Hilbert space may be understood as the completion of the space of $r$-times continuously differentiable functions under a suitable norm \citep{Adams:2003}. Therefore, it lends itself to slightly greater generality than the space of $r$-times continuously differentiable functions commonly employed in penalized spline and local polynomial estimation. Under this assumption, an intuitively appealing estimator for $\mu$ may be obtained by solving the smoothing spline problem
\begin{equation}
\label{eq:ROSM3}
\min_{f \in \mathcal{W}^{r, 2}([0,1])}\left[ \frac{1}{n} \sum_{i = 1}^n \frac{1}{m_i} \sum_{j=1}^{m_i} \rho\left( Y_{ij} - f(T_{ij}) \right) + \lambda \int_{0}^1 |f^{(r)}(t)|^2 dt \right],
\end{equation}
for some convex  nonnegative loss function $\rho$ satisfying $\rho(0)=0$ and a penalty parameter $\lambda >0$, whose positivity is needed to make the problem well-defined. 

The two extreme cases $\lambda \downarrow 0 $ and $\lambda \uparrow \infty$ reveal the compromise smoothing spline estimators aim to achieve. On the one hand, for $\lambda \downarrow 0 $ the solution becomes arbitrarily close to an natural spline of order $2r$ which interpolates the measurements $Y_{ij}$ . On the other hand, for $\lambda \to \infty$ the dominance of the penalty term in \eqref{eq:ROSM3} forces $||\widehat{\mu}^{(r)}_n||_2 =0$, where $\widehat{\mu}_n$ denotes a solution of \eqref{eq:ROSM3}.  A Taylor expansion with integral remainder of $\widehat{\mu}_n$ about $0$ shows that
\begin{equation*}
\widehat{\mu}_n(t) = P_r(t) + \int_{0}^1 \frac{\widehat{\mu}^{(r)}_n(x)}{(r-1)!} (t-x)_{+}^{r-1} dx,
\end{equation*}
where $P_{r}(t)$ is the Taylor polynomial of order $r$. But by the Schwarz inequality,
\begin{equation*}
\left| \int_{0}^1 \frac{\widehat{\mu}^{(r)}_n(x)}{(r-1)!} (t-x)_{+}^{r-1} dx \right| \leq  \frac{|t|^{r-1/2}}{(r-1)! (2r-1)^{1/2} } ||\mu^{(r)}||_2, 
\end{equation*}
so the Taylor remainder is equal to zero. This implies that for large $\lambda$, $\widehat{\mu}_n$ is roughly equal to its Taylor polynomial of degree at most $(r-1)$. For in between values of $\lambda$ the penalty functional can be viewed as providing a bound on how far $\mu$ is allowed to depart from a polynomial, as noted by \citet{Eubank:1999}.

While for general convex $\rho$-functions it cannot be assumed that the solution is unique \citep[Chapter 17]{Eg:2009}, for $\lambda>0$ and $\sum_{i=1}^n m_i \geq r$ it can be shown that a solution to the above variational problem may be found in the space of $2r$th order natural splines with knots at the unique $T_{ij}, j=1, \ldots, m_i, \ i = 1, \ldots, n$, see Theorem 1 of \citep{Kalogridis:2020}. Since this is a finite-dimensional space, with dimension equal to the number of unique knots, a solution to \eqref{eq:ROSM3} may be expediently found by solving a ridge-type regression problem, as discussed in detail in Section~\ref{sec:ROSM4} below.

The proposed estimator is a generalization of the robust smoothing spline introduced by \citet{Huber:1979} for the case of independent and identically distributed errors. It is clear that the squared loss $\rho(x) = x^2$ fulfils the requirements on $\rho$. In this case one recovers the least-squares smoothing spline with repeated measurements proposed by \citet{Rice:1991} and theoretically investigated by \citet{Cai:2011}. However, the benefit of the formulation in \eqref{eq:ROSM3} is that it includes loss functions that increase less rapidly as their argument increases in absolute value, so that better resistance towards outlying observations is achieved. Well-known examples of such loss functions are the $\mathcal{L}_{q}$ loss with $\rho_q(x) = |x|^{q}, 1\leq q \leq 2$, which includes both the absolute and square losses as special cases, and Huber's loss given by
\begin{equation*}
\rho_k(x) = \begin{cases} x^2/2, & |x| \leq k \\ k\left(|x|-k/2\right), & |x| >k,
\end{cases}\\
\end{equation*}
for some $k>0$ that controls the blending of square and absolute losses. For large $k$, $\rho_{k}(x)$ essentially behaves like a square loss while for small values of $k$, $\rho_k(x)$ provides a smooth approximation to the absolute loss.

For heterogeneous functional data a single location function may not be appropriate to summarize the data-generating process. In such situations quantiles or expectiles provide much richer information. Since we do not require the symmetry of the loss function in~\eqref{eq:ROSM3}, such versatile estimators may be readily incorporated into the present framework. Indeed, to compute conditional quantiles and expectiles one would only need to select the loss function according to $\rho_\tau(x) =  x(\tau -\mathcal{I}(x<0)) $ for $\tau \in (0,1)$ and $\rho_{\alpha}(x) = x^2/2|\alpha- \mathcal{I}(x <0)|$ for $\alpha \in (0,1)$ respectively. Furthermore, the asymptotic properties of quantile and expectile smoothing spline estimators are covered by the theory developed in Section~\ref{sec:ROSM3}.

\section{Asymptotic properties}
\label{sec:ROSM3}

We now examine the asymptotic properties of M-type smoothing spline estimators in two very different scenarios: the common design case and the independent design case. In the case of common design, the curves are recorded at the same locations, that is,
\begin{equation*}
T_{1j} = T_{2j}  = \ldots = T_{nj}, \quad (j = 1, \ldots, m),
\end{equation*}
while in the case of independent design we allow $T_{ij}$ to be random variables and consequently to be distinct across the subjects.

Central to our theoretical development in both types of designs is the theory of reproducing kernel Hilbert spaces. In particular, it is well-known that $\mathcal{W}^{r,2}([0,1])$ is a reproducing kernel Hilbert space but the reproducing kernel depends on the inner product. In this paper we make use of the inner product 
\begin{align*}
\langle f, g \rangle_{r, \lambda} = \langle f, g \rangle_2 + \lambda \langle f^{(r)}, g^{(r)} \rangle_2,
\end{align*}
for $f, g \in \mathcal{W}^{r, 2}([0,1])$, where $\langle \cdot, \cdot \rangle_2$ denotes the usual $\mathcal{L}^2([0,1])$ inner product. These quantities are well-defined and interestingly depend on the smoothing parameter $\lambda$, which typically varies with $n$. As we shall see, formulating our results in terms of $||\cdot||_{r, \lambda}$ instead of the commonly used $\mathcal{L}^2([0,1])$-norm will enable us to obtain uniform convergence rates as well as rates of convergence of the derivatives for a broad class of estimators, greatly extending the results of \citep{Cai:2011}. Our first result establishes the continuity of point evaluations under $||\cdot||_{r, \lambda}$ and characterizes the form of the resulting reproducing kernel. The interested reader is referred to \citep[Lemma 2.11, Chapter 13]{Eg:2009} for an alternative proof of the first part of Proposition~\ref{Prop:ROSM1} below.

\begin{prop}
\label{Prop:ROSM1}
Let $r \geq 1$ be an integer. For  all $\lambda \in (0,1]$ it holds that
\begin{itemize}
\item[(a)] There exists a constant $c_r$ depending only on $r$ such that, for all $f \in \mathcal{W}^{r,2}([0,1])$ and all $x \in [0,1]$,
\begin{align*}
|f(x)| \leq c_r \lambda^{-1/4r} ||f||_{r, \lambda}.
\end{align*}
\item[(b)] There exists a symmetric function $\mathcal{R}_{r, \lambda} : [0,1]^2 \to \mathbb{R}$, such that, for every $y \in [0,1]$, $x \mapsto \mathcal{R}_{r, \lambda}(x,y) \in \mathcal{W}^{r,2}([0,1])$ and for all $f \in \mathcal{W}^{r,2}([0,1])$, 
\begin{align*}
f(x) = \langle \mathcal{R}_{r, \lambda}(x,\cdot), f \rangle_{r, \lambda}.
\end{align*}
\item[(c)] There exists a sequence $\{\phi_j\}_j$ of uniformly bounded basis functions of $\mathcal{L}^2([0,1])$ not depending on $\lambda$, such that
\begin{align*}
\langle \phi_i, \phi_j \rangle_2 = \delta_{ij}, \quad \text{and} \quad \langle \phi_i, \phi_j \rangle_{r, \lambda} = (1+\lambda \gamma_i) \delta_{ij},
\end{align*}
where $\gamma_1 = \ldots = \gamma_r = 0$ and $C_1 j^{2r} \leq \gamma_j \leq  C_2  j^{2r}, j \geq r+1$, with $C_1, C_2 >0$. Furthermore, the reproducing kernel, $\mathcal{R}_{r, \lambda}(x,y)$, admits the representation
\begin{align*}
\mathcal{R}_{r, \lambda}(x,y) =  \sum_{j=1}^{\infty} \frac{\phi_j(x) \phi_j(y)}{1+\lambda \gamma_j},
\end{align*}
where the series converges absolutely and uniformly.
\end{itemize}
\end{prop}

The proposition requires that $\lambda \in (0,1]$, but this comes without any loss of generality, as in all our theoretical results below we will assume that $\lambda$ decays to zero as $n \to \infty$. The constant $c_r$ does not depend on $x$, thus Proposition \ref{Prop:ROSM1} defines a Sobolev embedding of $\mathcal{W}^{r,2}([0,1])$ into $\mathcal{C}([0,1])$ equipped with the standard sup-norm \citep{Adams:2003}. For general $r$, the explicit form of $\{\phi_j\}_j$ is difficult to determine but for $r=1$ we may use the results of \citep[Chapter 2]{Hsing:2015} to deduce that $\phi_1 = 1$ and
\begin{align*}
\phi_{j}(x) = 2^{1/2} \cos((j-1) \pi x), \quad j \geq 2,
\end{align*}
which is the standard cosine series of $\mathcal{L}^2([0,1])$, and $\gamma_j = (\pi j)^2, j \geq 2$. It is interesting to observe the rapid decay of the summands, which implies that the $\phi_j$ with larger frequencies contribute little to the value of the reproducing kernel at each point $(x,y)$. Of course, the precise rate of decay depends on the smoothing parameter with smaller values of $\lambda$ leading to a more slowly converging series. Another interesting observation is that while the $\phi_j$ are infinitely differentiable, the map $x \mapsto \mathcal{R}_{r, \lambda}(x,y)$ is not differentiable at $x= y$. This is not a contradiction to Proposition~\ref{Prop:ROSM1}, however, as that result only asserts the much weaker $x \mapsto \mathcal{R}_{r, \lambda}(x,y) \in \mathcal{W}^{1,2}([0,1])$. Figure~\ref{fig:ROSM1} plots this mapping for different values of $\lambda$ and $y$ in the left and right panels respectively.

\begin{figure}[H]
\centering
\subfloat{\includegraphics[width = 0.49\textwidth]{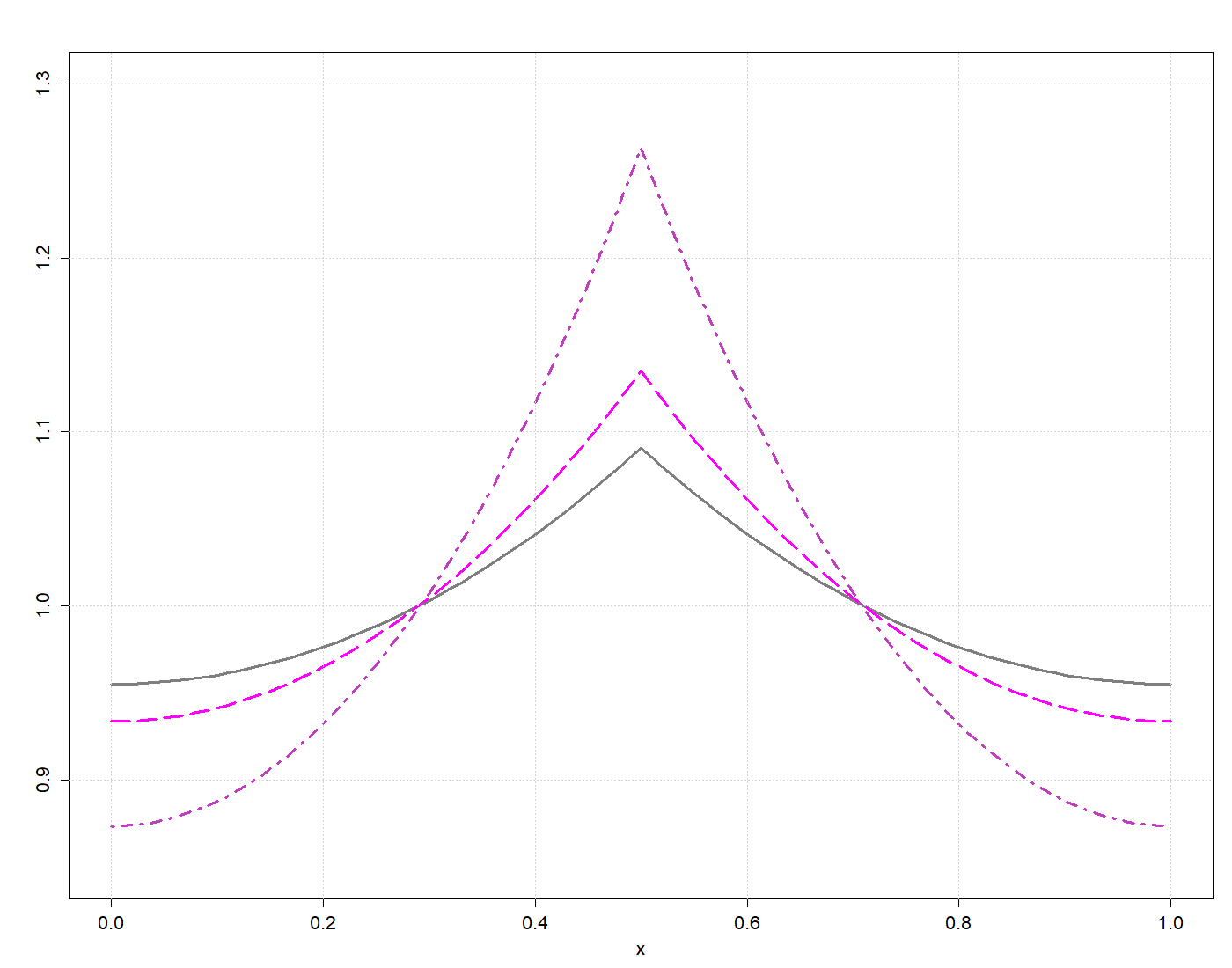}} \ 
\subfloat{\includegraphics[width = 0.49\textwidth]{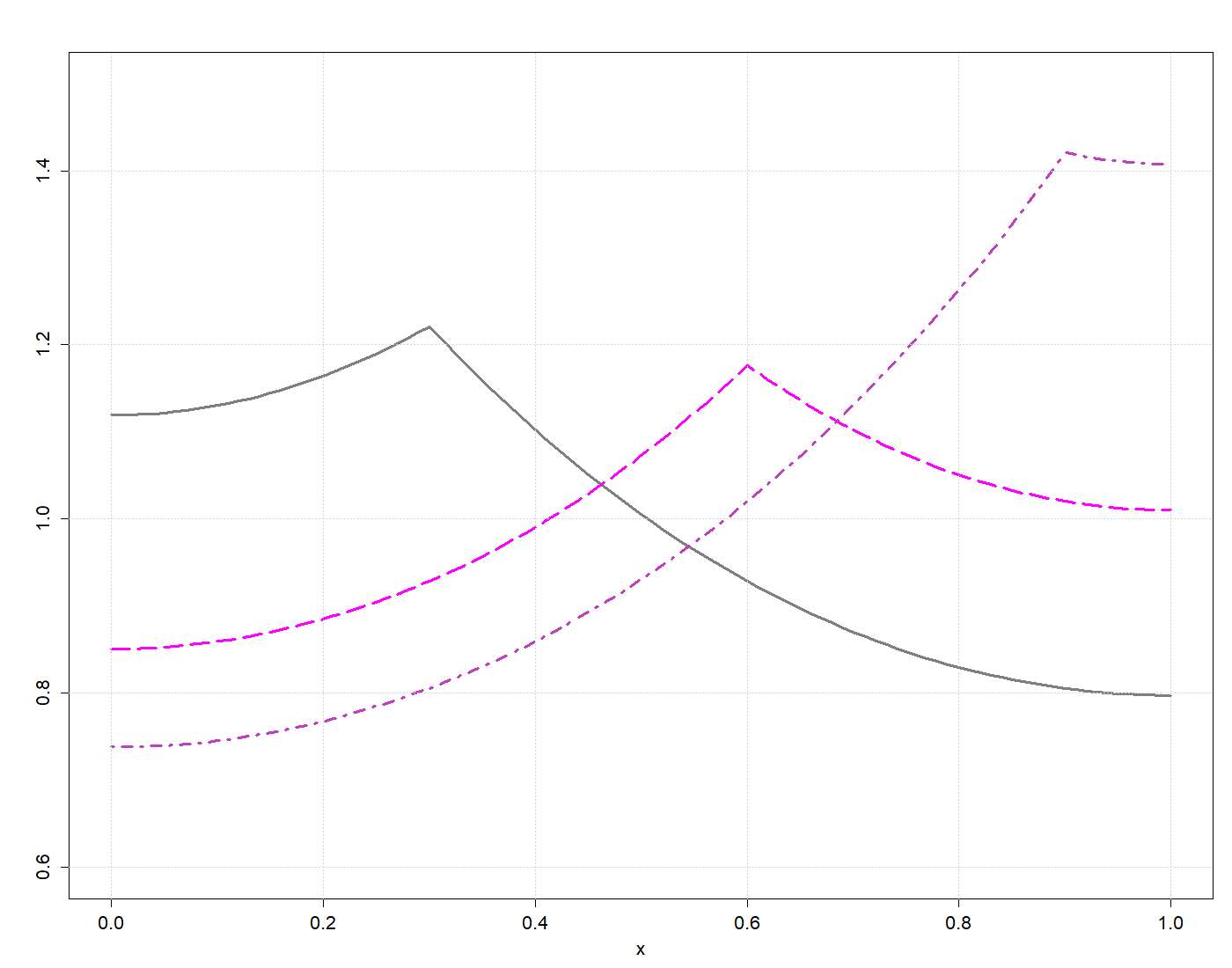}}
\caption{$\mathcal{R}_{1, \lambda}(x,y)$ as a function of $x$ for different values of $\lambda$ and $y$. In the left panel the lines ( \dotdash, \denselydashed, \full) correspond to $\lambda \in \{0.3, 0.6, 0.9\}$ with $y= 0.5$ while in the right panel the lines (\full, \denselydashed,  \dotdash) correspond to $y \in \{0.3, 0.6, 0.9\}$ with $\lambda = 0.5$. }
\label{fig:ROSM1}
\end{figure}

\subsection{Common design}

With this result at our disposal, we begin with the asymptotic study of M-type smoothing spline estimators in the common design. Notice that in this case the notation simplifies and the M-type estimator $\widehat{\mu}_n$ is defined according to
\begin{equation}
\label{eq:ROSM4}
\widehat{\mu}_n = \argmin_{f \in \mathcal{W}^{r, 2}([0,1])} \left[ \frac{1}{nm} \sum_{i = 1}^n \sum_{j=1}^m \rho\left( Y_{ij} - f(T_{j}) \right) + \lambda ||f^{(r)}||_2^2 \right].
\end{equation}

We require the following assumptions on the loss function, the error process and design points.
\begin{itemize}
\item[(A1)] $\rho$ is an absolutely continuous convex function on $\mathbb{R}$ with derivative $\psi$ existing almost everywhere.
\item[(A2)] The error processes $\epsilon_i(\omega, t) : (\Omega, \mathcal{A}, \mathbb{P}) \times [0,1] \to \mathbb{R},  i = 1, \ldots, n$, are independent and identically distributed.
\item[(A3)]  There exist finite constants $\kappa$ and $M_1$ such that for all $x\in \mathbb
{R}$ and $|y| < \kappa$,
\begin{align*}
|\psi(x+y)-\psi(x)| \leq M_1.
\end{align*}
\item[(A4)] There exists a finite constant $M_2$ such that
\begin{align*}
\sup_{j \leq m} \mathbb{E}\left\{ | \psi(\epsilon_{1j}+u) - \psi(\epsilon_{1j})|^2 \right\} \leq M_2 |u|,
\end{align*}
as $u \to 0$.
\item[(A5)] $\sup_{j \leq m} \mathbb{E}\{|\psi(\epsilon_{1j})|^2 \} < \infty$, $\mathbb{E}\{\psi(\epsilon_{1j}) \} = 0$  and there exist positive constants $\delta_j,  j = 1, \ldots, m$, such that $0 < \inf_{j \leq m}  \delta_j \leq \sup_{j \leq m} \delta_j < \infty$ and
\begin{equation*}
\sup_{j \leq m} \left| \mathbb{E}\left\{\psi(\epsilon_{1j}+u) \right\} - \delta_j u \right| = o(u),
\end{equation*}
as $u \to 0$.
\item[(A6)] The family of discretization points $T_j$ satisfies $T_0 := 0 < T_1 < T_2 <\ldots, < T_m < T_{m+1} := 1$ and 
\begin{equation*}
\max_{0 \leq j \leq m} |T_{j+1} - T_{j}| \leq c m^{-1},
\end{equation*}
for some constant $c>0$.
\end{itemize}

Assumption (A1) is standard in regression M-estimation, see, e.g., \citep{Wu:2007, Li:2011}. In our setting it ensures that (\ref{eq:ROSM3}) has a solution in $\mathcal{W}^{r,2}([0,1])$ and that all solutions will asymptotically be contained in the same ball of $\mathcal{W}^{r,2}([0,1])$, so that one does not need to worry about "bad" solutions in the limit. Assumption (A2) specifies the permissible error structure; it is substantially weaker than the assumption of \citet{Cai:2011}, as we do not require independence of the measurements errors $\zeta_{ij}, j = 1, \ldots, m$, associated with the $i$th subject. In our opinion, this as an important generalization because measurement errors occurring closely to each other in time or space are likely to be correlated in practice.

Assumptions (A3)--(A5) permit the use of discontinuous score functions by imposing some regularity on the error process and its finite-dimensional distributions instead. It is clear that assumptions (A3)--(A5) are satisfied for smooth $\psi$-functions. However, as \citet{Bai:1994} point out, these conditions hold quite generally. For example, they show that assumption (A4) may hold with the tighter $|u|^2$ on the right hand side, even if $\psi$ has jumps. On the whole, assumption (A5) ensures the Fisher-consistency of the estimators. The first part are two weak conditions that are satisfied, for example, if $\psi$ is bounded and odd and the error has a symmetric distribution about zero. The second part of assumption (A5) essentially requires that each of the functions $m_j(u) : = \mathbb{E}\{ \psi(\epsilon_{1j} + u) \}, j=1, \ldots, m$, is differentiable with strictly positive derivative at the origin. This is a necessary condition for the minimum to be well-separated in the limit. It is also not a stringent assumption and can be shown to hold for many popular loss functions. We now provide several examples to illustrate the broad applicability of this assumption.

\begin{example}[Squared loss] In this case $\psi(x) = 2x$ and the second part of assumption (A5) holds with $\delta_j = 2$, provided that $\mathbb{E}\{\epsilon_{1j} \} = 0$ for all $j \leq m$,  as in \citep{Cai:2011}.
\end{example}

\begin{example}[Smooth loss functions] All monotone everywhere differentiable $\psi$ functions with bounded second derivative $\psi^{\prime \prime}(x)$, such as $\rho(x) = \log(\cosh(x))$, satisfy the second part of assumption (A5) if
\begin{align*}
0<\inf_{j \leq m} \mathbb{E}\{\psi^{\prime}(\epsilon_{1j})\} \leq  \sup_{ j \leq m} \mathbb{E}\{\psi^{\prime}(\epsilon_{1j})\} <\infty.
\end{align*}
This is a straightforward generalization of the classical Fisher-consistency condition for spline estimation with smooth $\psi$-functions, see \citep{Cox:1983}.
\end{example}

\begin{example}[Check loss]
First, consider the absolute loss for which $\psi(x)=\sign(x)$. If each $\epsilon_{1j}, j = 1, \ldots, m,$ has a distribution function $F_j$ with positive density $f_j$ on an interval about zero, then
\begin{equation*}
\mathbb{E}\{ \sign(\epsilon_{1j}+u) \} = 2f_j(0)u + o(u), \quad \text{as} \quad u \to 0,
\end{equation*}
so that the assumption in (A5) holds with $\delta_j = 2f_j(0)$. This easily generalizes  to the check loss, provided that in this case one views the regression function $\mu$ as the $\tau$-quantile function, that is, $\Pr(Y_{ij} \leq \mu(T_{j}) ) = \tau$, see \citep{Wei:2006}.
\end{example} 

\begin{example}[Huber loss]
Now $\psi_k(x) = x \mathcal{I}( |x| \leq k) + k \sign(x) \mathcal{I}(|x|>k)$ for some $k>0$. Assuming that each $F_j, j = 1, \ldots, m,$ is absolutely continuous and symmetric about zero we have that
\begin{equation*}
\mathbb{E}\{ \psi_k(\epsilon_{1j}+u) \} = \{ 2 F_j(k)-1\}u + o(u), \quad \text{as} \quad u \to 0.
\end{equation*}
It is clear that the assumption in (A5) holds in this case with $\delta_j = 2 F_j(k)-1$ provided that $\min_{j} F_j(k) > 1/2$. 
\end{example}

\begin{example}[$L_{q}$ loss with $q \in (1,2)$] Clearly, $\psi_q(x) = q |x|^{q-1} \sign(x)$ and if we assume that each $F_j$ is symmetric about zero, $\mathbb{E}\{| \epsilon_{1j}|^{q-1}\} < \infty$ and $\mathbb{E}\{|\epsilon_{1j}|^{q-2} \}<\infty$, then
\begin{equation*}
\mathbb{E}\{\psi_q(\epsilon_{1j}+t) \} = q(q-1)\mathbb{E}\{|\epsilon_{1j}|^{q-2} \} u + o(u), \quad \text{as} \quad u \to 0,
\end{equation*}
see \citep{Arcones:2000}. The latter expectation is finite, if, e.g., $F_j$ possesses a Lebesgue density $f_j$ that is bounded at an interval about zero.
\end{example}

\begin{example}[Expectile loss]
As an alternative to the check loss, consider the continuously differentiable expectile loss $\rho_{\alpha}(x) = x^2/2(| \alpha - \mathcal{I}(x\leq 0)|)$ with $\alpha \in (0,1)$, such that $\psi_{\alpha}(x) = (1-\alpha)x \mathcal{I}(x \leq 0) + \alpha x\, \mathcal{I}(x>0)$. Assuming that there is an interval about the origin in which no $F_j$ has atoms we have
\begin{equation*}
\mathbb{E}\{\psi_{\alpha}(\epsilon_{1j}+u) \} =  \{\alpha + (1- 2 \alpha) F_j(0)\}u + o(u), \quad \text{as} \quad u \to 0.
\end{equation*}
The term in curly braces is positive for each $j$ and $\alpha \in (0,1)$. Therefore, the assumption in (A5) holds with $\delta_j = \alpha + (1- 2 \alpha) F_j(0)$.
\end{example}

Condition (A6) ensures that the sampling points are unique and observed at a sufficiently regular grid. It is again a weak assumption that can be shown to hold, for example, if $T_j = 2j/{(2m+2)}, j = 1, \ldots, m$, and many other designs. See \citep{Cai:2011} and \citep{Xiao:2020} for close variants of this condition.

With these assumptions we can now state our first main asymptotic result. To ensure consistency of the estimators we assume that the discretization points become more numerous as the sample size grows larger. That is, we require that $m$ depends on $n$ and $m(n) \to \infty$ as $n\to \infty$. This is a natural requirement in order to examine consistency in norms that involve integrals. Also the penalty parameter $\lambda$ needs to depend on $n$. To avoid making the notation too heavy, we often not explicitly indicate the dependencies of $m$ and $\lambda$ on $n$.

\begin{theorem}
\label{Thm:ROSM1}
Consider model \eqref{eq:ROSM2} and the M-type smoothing spline estimator $\widehat{\mu}_n$ satisfying \eqref{eq:ROSM4}. Assume that assumptions (A1)-(A6) hold. Moreover, assume that for $n \to \infty$, $m$ and $\lambda$ vary in such a way that $m \to \infty$, $\lambda \to 0 $, $n \lambda^{1/2 +(1 - 1/8r)/r } \to \infty$ and \quad $\liminf_{n\to \infty} m \lambda^{1/{2r}}\geq 2 c_r$, where $c_r$ is the constant of Proposition \ref{Prop:ROSM1}. Then,
 every sequence of M-type smoothing spline estimators $\widehat{\mu}_n$ satisfies
\begin{equation*}
||\widehat{\mu}_n - \mu||_{r, \lambda}^2 = O_P\left( n^{-1} +m^{-2r} +  \lambda \right).
\end{equation*}
\end{theorem}
The asymptotic behaviour of the smoothing spline estimator in this functional setting given in Theorem~\ref{Thm:ROSM1} is very different from the asymptotic behaviour of M-type smoothing splines in nonparametric regression with i.i.d. errors as studied by~\citep{Cox:1983, Kalogridis:2020}. In fact, the only similarity is the role of $\lambda$ which needs to tend to zero in order for the regularization bias to become negligible. The limit conditions of the theorem are satisfied if $\lambda \sim C (m^{-2r} + n^{-1}) $ for sufficiently large $C$ and $r \geq 2$. They can also be satisfied for $r=1$, i.e., the case of linear smoothing splines, provided that $m$ is "small" relative to $n$, which precludes densely observed functional data. For $r=1$ this rate of convergence can be attained without any conditions on $m$, provided that one uses a smoother  (Lipschitz-continuous) score function $\psi$. Thus, smoother score functions allow more flexibility in the choice of $r$ and the penalty parameter $\lambda$.

It is interesting to observe that for $\lambda \asymp (m^{-2r} + n^{-1})$ the $\mathcal
{L}^2([0,1])$-error decays like
\begin{equation*}
||\widehat{\mu}_n - \mu||_{2}^2 \leq ||\widehat{\mu}_n - \mu||_{r, \lambda}^2 = O_P\left( n^{-1} + m^{-2r} \right),
\end{equation*}
which is the optimal rate of convergence in the common design \citep{Cai:2011}. This rate of convergence implies that a phase transition between the two sources of error occurs at $ m \asymp n^{1/2r}$. Indeed, if $m>>n^{1/{2r}}$ then the asymptotic error behaves like $n^{-1}$, which is the rate of convergence for many functional location statistics that assume the curves are observed entirely, see \citep{Horv:2012} and \citep{Gervini:2008} for the functional mean and median respectively. On the other hand, if $m = O(n^{1/{2r}})$ then the asymptotic error behaves like $m^{-2r}$, which is the error associated with piecewise polynomial interpolation of a $\mathcal{W}^{r,2}([0,1])$-function \citep{Devore:1993}. This	 result confirms the intuitive notion that as long as the grid is dense, the error will decay with the parametric rate $n^{-1}$. However, Theorem \ref{Thm:ROSM1} also shows that for sparsely observed data the discretization error cannot be ignored in the way many currently available resistant estimation procedures require.

Based on Theorem \ref{Thm:ROSM1} we can now also obtain rates of convergence for the derivatives $\widehat{\mu}_n^{(s)}, s = 1, \ldots, r-1$ and tightness of $\widehat{\mu}_n^{(r)}$ in the classical $\mathcal{L}^2([0,1])$-metric. These are summarized in Corollary \ref{Cor:ROSM1} below.
\begin{cor}
\label{Cor:ROSM1}
Under the assumptions of Theorem \ref{Thm:ROSM1}, it holds that 
\begin{equation*}
||\widehat{\mu}_n^{(s)}-\mu^{(s)}||_2^2 = \lambda^{-s/r} O_P\left( n^{-1} + m^{-2r} + \lambda \right),
\end{equation*}
for all $s \leq r$.
\end{cor}
Since, by assumption, $\lambda \to 0$ as $n \to \infty$, the corollary implies that derivatives of higher order are more difficult to estimate, ceteris paribus. This is not surprising as differentiability is a local property and one is expected to need a larger sample size and higher grid resolution in order to examine local properties with the same degree of precision.

Convergence in $\mathcal{L}^2([0,1])$ is, in general, too weak to imply pointwise convergence, but the stronger mode of convergence obtained in Theorem~\ref{Thm:ROSM1} and Proposition~\ref{Prop:ROSM1} allow us to derive a uniform rate of convergence as follows.

\begin{cor}
\label{Cor:ROSM2}
Under the assumptions of Theorem \ref{Thm:ROSM1}, it holds that 
\begin{equation*}
\sup_{t \in [0,1]} |\widehat{\mu}_n(t) - \mu(t)| = \lambda^{-1/4r}  O_P\left( n^{-1} + m^{-2r} + \lambda \right)^{1/2}.
\end{equation*}
\end{cor}
\noindent
Besides its intrinsic interest, Corollary~\ref{Cor:ROSM2} may be viewed as the first step towards establishing uniform confidence bands for $\mu$. We aim to study this problem in detail in forthcoming work.

\subsection{Independent design}
We now turn to the problem of general location estimation in the independent design case, i.e. when trajectories are observed at different points and their number can vary from trajectory to trajectory. A convenient way to model this difference is by treating the $T_{ij}, j =1, \ldots, m_i, i = 1, \ldots, n$, as independent and identically distributed random variables. 

The assumptions that we need in order to obtain the main result in this case are for the most part straightforward generalizations of the assumptions governing the common design case. We denote the set of sampling points $\{T_{ij}\}_{i,j}$ by $\mathcal{T}$.

\begin{itemize}
\item[(B1)] (A1).
\item[(B2)] (A2).
\item[(B3)] (A3).
\item[(B4)] There exists a constant $M_2$ such that
\begin{align*}
\sup_{i, j} \mathbb{E}\left\{ |\psi(\epsilon_{ij}+u) - \psi(\epsilon_{ij})|^2 | \mathcal{T} \right\}\leq M_2 |u|,
\end{align*}
as $u \to 0$.
\item[(B5)] $\mathbb{E}\{||\psi(\epsilon_{1})||^2_2\} < \infty$, $\mathbb{E}\{\psi(\epsilon_{ij})| 	\mathcal{T} \} = 0, j = 1, \ldots, m_i, i = 1, \ldots, n$, and there exist positive constants $\delta_{ij}$ such that $0 <\inf_{i, j} \delta_{ij} \leq \sup_{i, j} \delta_{ij} < \infty $ and
\begin{equation*}
\sup_{i,j} \left|\mathbb{E}\left\{\psi(\epsilon_{ij}+u) | 	\mathcal{T} \right\} - \delta_{ij}u \right| = o_P(u), \quad \text{as} \quad u \to 0.
\end{equation*}
\item[(B6)] The sampling points $T_{ij}, j=1, \ldots, m_i, i = 1, \ldots, n$, are independent and identically distributed random variables with Lebesgue measure on $[0,1]$. They are also independent of the error processes $\epsilon_i(t), i=1, \ldots, n$.
\end{itemize}

Assumptions (B4)--(B5) are adaptations of assumptions (A4)--(A5) to the case of distinct sampling points. Assumption (B6) ensures that the discretization points are well-spread throughout the $[0,1]$-interval. The assumption of uniformity can be marginally weakened to requiring that the density is bounded away from zero and infinity on $[0,1]$ at the cost of heavier notation and lengthier derivations in the proofs. 

A quantity of importance for the asymptotic properties of M-type estimators is the harmonic mean $m$ of the $m_i$, that is,
\begin{align*}
m := \left( \frac{1}{n} \sum_{i=1}^n \frac{1}{m_i} \right)^{-1}.
\end{align*}
We may think of $m$ as a measure of how dense the data is: lower values indicate sparse functional data while larger values indicate more densely sampled functional data. As in \citep{Xiao:2020} the $m_i$ are deterministic in our analysis. The more complex situation of random $m_i$ from a distribution with support on the integers will be considered in future work.

\begin{theorem}
\label{Thm:ROSM2}
Consider model~(\ref{eq:ROSM2}) and the M-type smoothing spline estimator $\widehat{\mu}_n$ solving~\eqref{eq:ROSM3}. Assume that assumptions (B1)--(B6) hold. Moreover, assume that for $n \to \infty$, $\lambda$ varies in such a way that $\lambda \to 0$ and $ n \lambda^{1/2 +(1-1/8r)/r} \to \infty $. Then, every sequence of M-type smoothing spline estimators $\widehat{\mu}_n$ satisfies
\begin{equation*}
||\widehat{\mu}_n - \mu||_{r, \lambda}^2 = O_P\left( (n m \lambda^{1/2r})^{-1} + \lambda + n^{-1} \right).
\end{equation*}
\end{theorem}

The limit conditions of Theorem~\ref{Thm:ROSM2} parallel those of Theorem~\ref{Thm:ROSM1} and are given in such a way that they do not involve $m$, which in the case of an independent design cannot be assumed to tend to infinity. Taking $\lambda \asymp (mn)^{-2r/{2r+1}}$ ensures that the limit assumptions are satisfied and we are lead to
\begin{equation*}
||\widehat{\mu}_n - \mu||_{2}^2 = O_P\left( (mn)^{-2r/(2r+1)} + n^{-1} \right),
\end{equation*}
which are the optimal rates of convergence for the independent design \citep{Cai:2011}. Similarly to the common design case, the phase transition occurs at $m \asymp n^{1/{2r}}$, although in the present case the conclusions are more subtle. For $m>>n^{1/{2r}}$ it can again be seen that the asymptotic error behaves like $n^{-1}$. However, a rather interesting phenomenon occurs when $m =  O(n^{1/{2r}})$. In this case, the asymptotic error behaves like $(mn)^{-2r/{(2r+1)}}$, which would have been the optimal nonparametric rate of convergence had we possessed $mn$ independent observations at $mn$ different sites \citep{Cox:1983, Kalogridis:2020}. Thus, in these extreme cases the estimator behaves either as if the curves were fully observed, or as if they were only observed at $n m$ points in an independent manner.

Corollaries~\ref{Cor:ROSM1} and \ref{Cor:ROSM2} regarding derivative estimation and uniform convergence follow in the same manner as previously through Sobolev embeddings. It is curious that while derivatives in many instances provide insightful tools for the analysis of functional data, they have received little theoretical attention in the sparse setting and, to the best of our knowledge, all existing results, such as \citep{Liu:2009}, concern least-squares estimators. The present methodology provides not only rates of convergence for a broad class of estimators, but also an expedient way of computing these derivatives through the use of the B-spline basis, which we outline next.

\section{Computation and smoothing parameter selection}
\label{sec:ROSM4}

As discussed in Section \ref{sec:ROSM2}, there exists at least one solution of \eqref{eq:ROSM3} in the space of natural splines of order $2r$ with knots at the unique $T_{ij}$. Thus we may restrict attention to the linear subspace of natural splines for the computation of the estimator. Assume for simplicity that all $T_{ij}$ are distinct and let $a = \min_{ij} T_{ij}>0$ and $b= \max_{ij} T_{ij}<1$. Then the natural spline has $v := \sum_{i=1}^n m_i -2$ interior knots and we may write $\mu(t) = \sum_{k = 1}^{v + 2r} \mu_k B_{k}(t)$ where the $B_k(\cdot)$ are the B-spline basis functions of order $2r$ supported by the knots at the interior points $T_{ij}$ and the $\mu_k$ are scalar coefficients.  It is well-known that the space of natural splines with $v$ interior knots is $v$-dimensional due to the boundary conditions, but for the moment we operate in the larger $(v + 2r)$-dimensional spline subspace. This is computationally convenient due to the local support of B-splines. As will be explained below, the penalty automatically imposes the boundary conditions \citep[see also][pp. 161-162]{Hastie:2009}.

With the introduction of the B-spline basis, to compute the estimator it suffices to find the vector $\widehat{\boldsymbol{\mu}} = (\widehat{\mu}_1, \ldots, \widehat{\mu}_{v+2r})^{\top}$ such that
\begin{align}
\label{eq:ROSM5}
\widehat{\boldsymbol{\mu}} = \argmin_{\boldsymbol{\mu} \in \mathbb{R}^{v+2r}} \left[  \frac{1}{n} \sum_{i=1}^n \frac{1}{m_i} \sum_{j=1}^{m_i} \rho\left(Y_{ij} - \mathbf{B}_{ij}^{\top} \boldsymbol{\mu} \right) + \lambda \boldsymbol{\mu}^{\top} \mathbf{P} \boldsymbol{\mu} \right]
\end{align}
for $\mathbf{B}_{ij} = ( B_1 (T_{ij}), \ldots, B_{v+2r}(T_{ij}) )^{\top}$ and $\mathbf{P} = \langle B_k^{(r)}, B_l^{(r)} \rangle_2, k, l =  1, \ldots, v+2r$. Note that the derivatives of B-splines may be written in terms of B-splines of lower order by taking weighted differences of the coefficients \citep[p. 117]{DB:2001}. Initially, it may seem that this formulation ignores the boundary constraints that govern natural splines but it turns out the penalty term automatically imposes them. The reasoning is as follows: if that were not the case, it would always be possible to find a $2r$th order natural interpolating spline that leaves the first term in \eqref{eq:ROSM5} unchanged, but because it is a polynomial of order $r$ outside of $[a,b] \subset [0,1]$ the penalty semi-norm  would be strictly smaller. Hence the minimizer of \eqref{eq:ROSM5} incorporates the boundary conditions.

For $\rho$-functions whose derivative exists everywhere, the solution to \eqref{eq:ROSM5} may be expediently found through minor modification of the penalized iteratively reweighted least-squares algorithm, see e.g., \citep{Maronna:2011}. The algorithm consists of solving a weighted penalized least-squares problem at each iteration until convergence, which is guaranteed irrespective of the starting values and yields a stationary point of \eqref{eq:ROSM3}, under mild conditions on $\rho$ that include the boundedness of $\rho^{\prime}(x)/x$ near zero~\citep{Huber:2009}.  These conditions are satisfied for smooth loss functions as well as the Huber and expectile losses, but are not satisfied for the popular quantile loss. Nevertheless, the easily implementable recipe of \citep{Nychka:1995} may be used in order to obtain an approximate solution of \eqref{eq:ROSM5}. In particular, in the algorithm the loss function can be replaced by the smooth approximation
\begin{equation*}
\tilde{\rho}_{\tau}(x) = \begin{cases} \rho_{\tau}(x) &  |x| \geq \epsilon \\
\tau x^2/\epsilon & 0 \leq x < \epsilon \\
(1-\tau)x^2/\epsilon & -\epsilon <  x \leq 0,
\end{cases}
\end{equation*}
for some small $\epsilon>0$. Whenever such a modification of the objective function is not feasible, we recommend utilizing a convex-optimization program in order to identify a minimizer of \eqref{eq:ROSM5}, as given, for example, by \citet{Fu:2020}.

To determine the penalty parameter $\lambda$ in a data-driven way we propose to select $\lambda$ that minimizes the generalized cross-validation (GCV) criterion
\begin{equation}
\GCV(\lambda) = \frac{ \sum_{i=1}^n m_i^{-1} \sum_{j=1}^{m_i} W(r_{ij}) |r_{ij}|^2 }{{(1- n^{-1} \sum_{i=1}^n m_i^{-1} \sum_{j=1}^{m_i} h_{ij} )^2 }}  ,
\label{GCVcrit}
\end{equation}
where $r_{ij}$ is the $ij$th residual, $W(r_{ij}) = \rho^{\prime}(r_{ij})/r_{ij}$ and the $h_{ij}$ are measures of the influence of the $ij$th observation, which can for instance be obtained from the diagonal of the weighted hat-matrix obtained upon convergence of the iterative reweighted least-squares algorithm. The GCV criterion employed herein is a generalization of the criterion proposed by \citet{Cunningham:1991}. Implementations and illustrative examples of the Huber and quantile-type smoothing spline estimators are available in \url{https://github.com/ioanniskalogridis/Robust-optimal-estimation-of-functional-location}.

\section{Finite-sample performance}
\label{sec:ROSM5}

We now compare the numerical performance of the proposed robust M-type smoothing spline estimator to several competitors. In particular, we include the least-squares smoothing spline estimator of \citep{Rice:1991, Cai:2011}, the robust regression spline estimator of \citep{Lima:2019}, the least-squares local linear estimator of \citep{Yao:2005, Degras:2011} and the robust functional M-estimator of \citep{Sinova:2018} in the comparison. For simplicity and for ease of comparison with the functional M-estimator of \citep{Sinova:2018}, which at the very least requires densely sampled trajectories, we consider the case of common design. First, we briefly review the construction of the competing estimators.

Let $K:\mathbb{R} \to [0,1]$ denote a symmetric nonnegative kernel (weight) function on $\mathbb{R}$ that is  Lipschitz continuous. Then, for any given point $t$  the local linear estimator of \citep{Yao:2005,Degras:2011} estimates $\mu(t)$ by $\widehat{\mu}_{LSLP}(t)=\widehat{\beta}_0$, where $(\widehat{\beta}_0, \widehat{\beta}_1)$ solves the local linear problem
\begin{equation*}
\min_{(\beta_0, \beta_1) \in \mathbb{R}^2 } \sum_{j=1}^m|\bar{Y}_j - \beta_0 - \beta_1(T_j-t)|^2 K\left( \frac{T_j-t}{h} \right).
\end{equation*}
Here, $\bar{Y}_j$ denotes the average of the responses at $T_j$ and $h$ denotes the smoothing parameter (i.e., the bandwidth). In our implementation we use the Gaussian kernel for $K$. The bandwidth $h$ is chosen by the direct plug-in method, which essentially works by replacing the unknown functionals in the asymptotic expression of the optimal bandwidth with Gaussian kernel estimates. 

For fully observed i.i.d. curves $X_1, \ldots, X_n$, \citet{Sinova:2018} propose to estimate $\mu(t)$ by $\widehat{\mu}_{HF}$ such that 
\begin{equation}
\label{eq:ROSM6}
\widehat{\mu}_{HF} = \argmin_{f \in \mathcal{L}^2([0,1])} \frac{1}{n} \sum_{i=1}^n \rho\left(||X_i - f||_{2} \right),
\end{equation}
with $\rho$ a real-valued loss function.
The well-known spatial median can be recovered from \eqref{eq:ROSM6} by setting $\rho(x) =x$. 
For this comparison we use the Huber loss function with tuning parameter equal to $0.70$, which corresponds to 85\% efficiency in the Gaussian location model. 
To apply the estimator in our setting,  $\mathcal{L}^2([0,1])$-norms and inner products were computed with trapezoidal Riemann approximations wherever necessary.

For discretely but commonly observed functional data, \citet{Lima:2019} proposed estimating $\mu$ by $\sum_{j=1}^{p} \widehat{\mu}_{j} B_{j}$ where $\boldsymbol{\mu} = (\mu_1, \ldots, \mu_p)^{\top}$ satisfies
\begin{align*}
\widehat{\boldsymbol{\mu}} = \argmin_{\boldsymbol{\mu} \in \mathbb{R}^p} \sum_{i=1}^n \sum_{j=1}^m \rho\left(Y_{ij} - \mathbf{B}_j^{\top}\boldsymbol{\mu}\right),
\end{align*}
with $B_1, \ldots, B_p$ B-spline basis functions. The loss function $\rho$ is taken to be the Huber loss with tuning again equal to 0.70. The vector $\mathbf{B}_j$ stands for a collection of cubic B-spline basis functions defined on equidistant knots evaluated at $T_j$, i.e., $\mathbf{B}_j = (B_1(T_j), \ldots, B_p(T_j))^{\top}$. Motivated by asymptotic considerations, the authors propose to choose the dimension of the approximating spline subspace as $p = \max\{4,[0.3 n^{1/4} \log(n)]\}$.

In our simulation experiments we are particularly interested in the effect of the sampling frequency, the level of noise and the effect of atypical observations on the estimates. To reflect these considerations we have generated independent and identically distributed curves according to the Karhunen-Loève expansion
\begin{equation*}
X(t) = \mu(t) + 2^{1/2}\sum_{k=1}^{50} W_k \frac{\sin((k-1/2)\pi t)}{(k-1/2)\pi},
\end{equation*}
for $t \in (0,1)$. The mean function $\mu(t)$ is either $\sin(6 \pi t)(t+1)$ or $3\exp[-(0.25-t)^2/ 0.1]$. The $W_k$ are independent and identically distributed random variables following the $t_{5}$ distribution.  Using the $t_{5}$ distribution instead of the commonly used standard Gaussian distribution implies that some of the curves exhibit some outlying behaviour, so that our samples may contain observations with large influence.

Once the curves have been generated we discretize them in $m \in \{20,50\}$ equispaced points $0<T_{1}< \ldots< T_{m}<1$ and generate the noisy observations according to 
\begin{equation*}
Y_{ij} = X_i(T_j) + \sigma \zeta_{ij}, , \quad (j = 1, \ldots, m; \ i=1, \ldots, n),
\end{equation*}
with $\sigma>0$ a constant that controls the level of noise in the data. We consider the values $\sigma \in \{0.2,0.5, 1\}$, reflecting low, medium and high levels of noise, respectively. The errors $\zeta_{ij}$ are i.i.d. random variables generated according to each of the following distributions: (i) standard Gaussian, (ii) $t_3$ distribution, (iii) skewed $t_3$ distribution with non-centrality parameter equal to 0.5, (iv) a convex combination of independent Gaussians with means equal to zero, variances equal to $1$ and $9$ respectively and weights equal to $0.85$ and $0.15$, respectively and (v) Tukey's Slash distribution, i.e., the quotient of independent standard Gaussian and $(0,1)$-uniform random variables. Figure \ref{fig:ROSM2} presents examples of a typical set of observations under standard Gaussian errors and small noise with the first and second mean-function, respectively.

\begin{figure}[H]
\centering
\subfloat{\includegraphics[width = 0.49\textwidth]{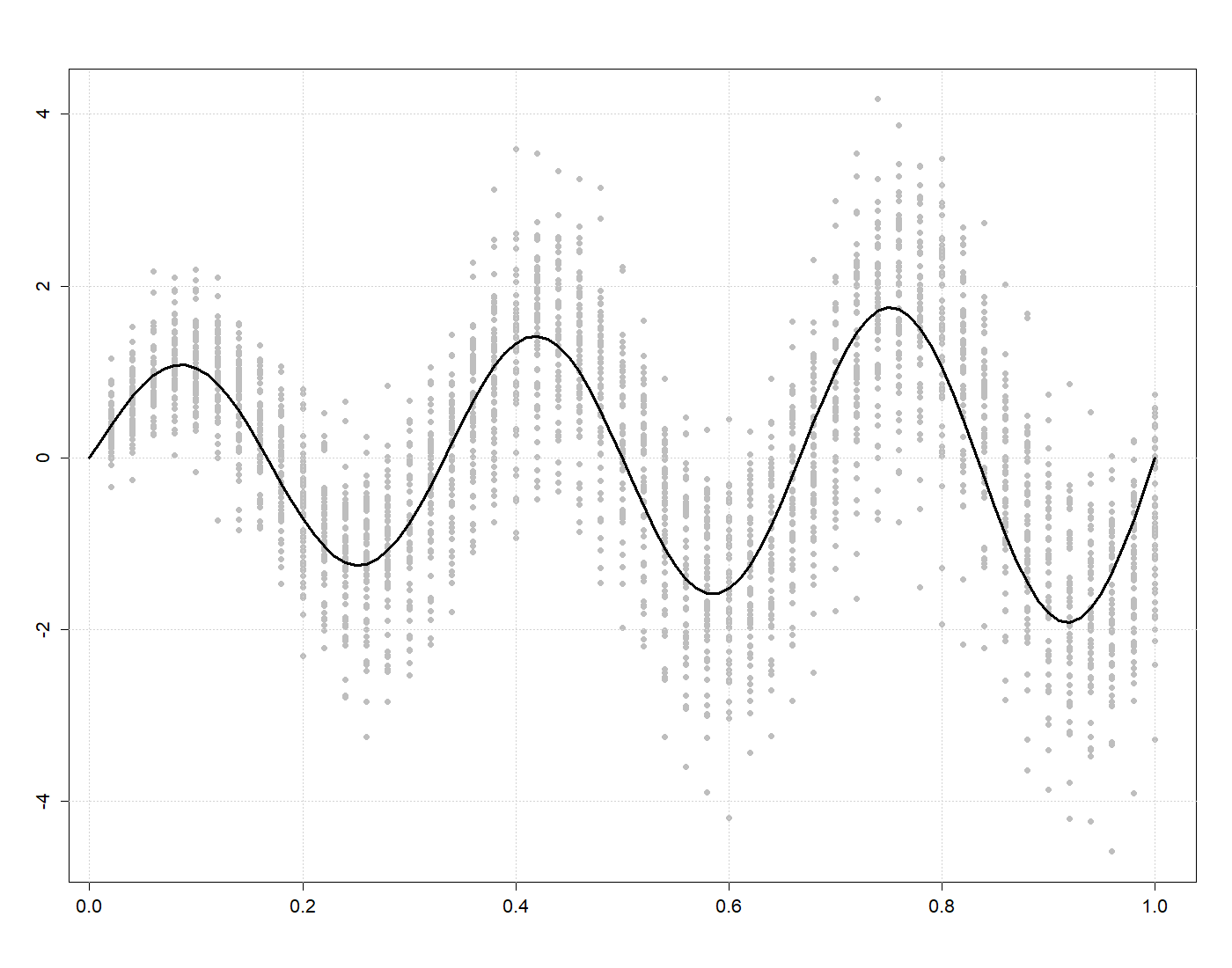}} \ 
\subfloat{\includegraphics[width = 0.49\textwidth]{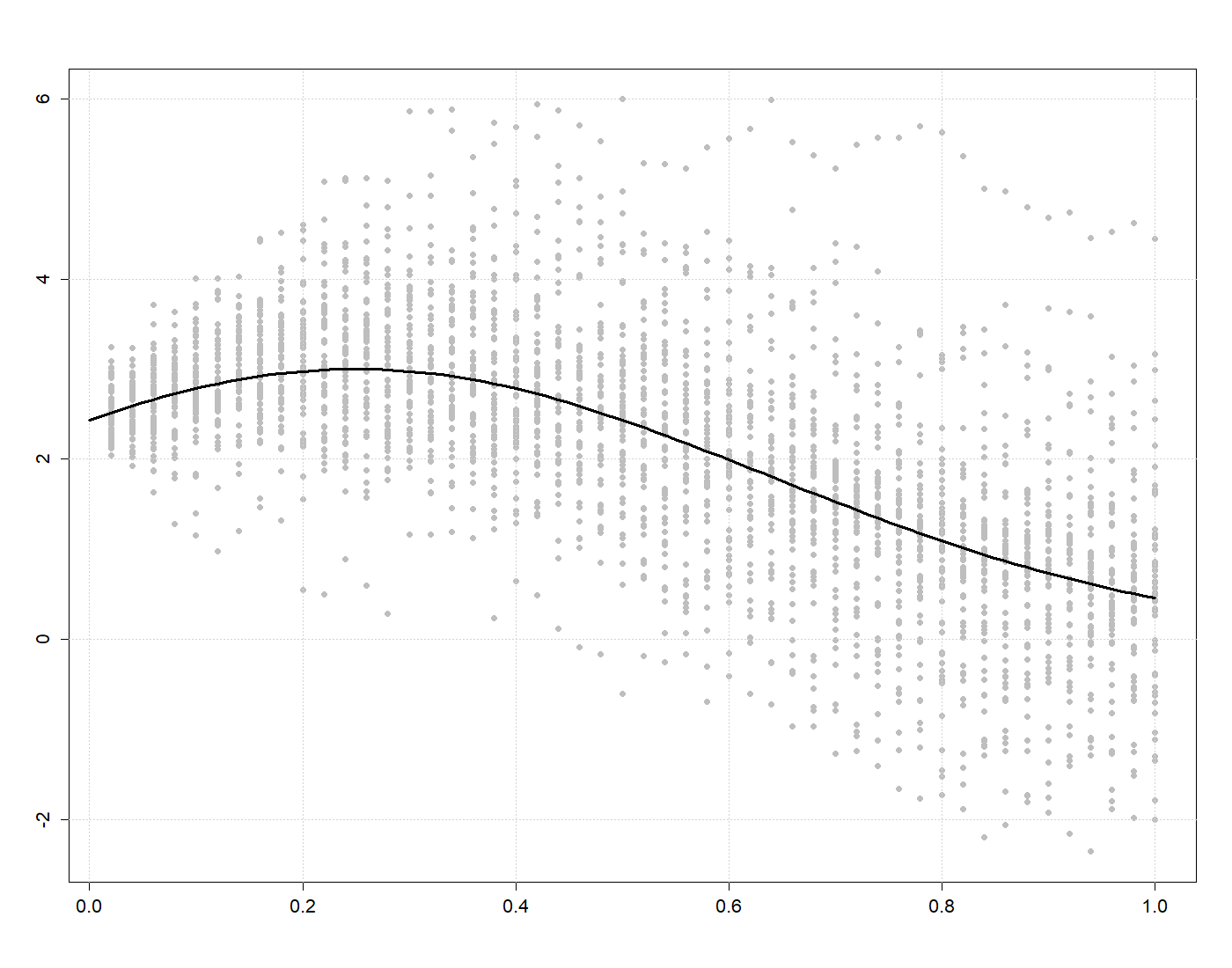}}
\caption{Two representative samples with mean function equal to $\sin(6 \pi t)(t+1)$ (left) and $3\exp[-(t-0.25)^2/0.1]$ (right), $\sigma = 0.2$ and standard Gaussian errors $\zeta_{ij}$. The mean functions are depicted as solid black lines.}
\label{fig:ROSM2}
\end{figure}

For this comparison we consider the M-type smoothing spline estimator based on the Huber function, also with tuning constant equal to 0.70, and henceforth denote this estimator by $\widehat{\mu}_{HSP}$. The least-squares smoothing spline estimator described by \citet{Cai:2011} is denoted by $\widehat{\mu}_{LSSP}$. These authors do not recommend a method to select the penalty parameter in finite samples. Hence, similarly as for the M-type smoothing spline estimator we use the GCV criterion in~(\ref{GCVcrit}). Both spline estimators were fitted with $r=2$ resulting in cubic splines. All estimators were implemented in the freeware $\texttt{R}$ \citep{R}. For the local linear estimator we used the package \texttt{KernSmooth} \citep{KernSmooth}, for the regression spline estimator we used a convex optimization routine provided by \citet{Fu:2020}, while the functional M-estimator was implemented according to the algorithm provided by \citet{Sinova:2018}. 

\begin{table}[H]
\centering
\resizebox{\columnwidth}{!}{%
\begin{tabular}{ccccccccccccc}
& \multicolumn{1}{c}{} &  \multicolumn{1}{c}{}  & \multicolumn{2}{c}{$\widehat{\mu}_{LSSP}$} & \multicolumn{2}{c}{$\widehat{\mu}_{HSP}$} & \multicolumn{2}{c}{$\widehat{\mu}_{LSLP}$} & \multicolumn{2}{c}{$\widehat{\mu}_{HF}$}  & \multicolumn{2}{c}{$\widehat{\mu}_{RS}$} \\ \\[-2ex]
$\sigma$ & $ m$ & Dist. & Mean & SE & Mean & SE & Mean & SE & Mean & SE  & Mean & SE   \\ \\
\multirow{10}{*}{$0.2$} & \multirow{5}{*}{20} & Gaus.  & 0.018 &  0.0006 & 0.0172 & 0.0005 &  0.2020 & 0.001 & \textbf{0.0132} & 0.0004 & 0.9216 &  0.0005 \\
& & $t_3$ & 0.0194  & 0.00061 & 0.0181 & 0.0004 & 0.2033 & 0.001 &  \textbf{0.0149} & 0.0004 & 0.9221 & 0.0005  \\
& & $st_{3, 0.5}$  & 0.0300 & 0.0009 & 0.0282 & 0.0007 & 0.190 & 0.0009 &  \textbf{0.0264} & 0.0008 & 0.9322 & 0.0008 \\
& & M. Gaus. &  0.0249 &  0.0006 &  \textbf{0.0206} & 0.0005 & 0.213 &  0.0016 & 0.0212 & 0.0005 & 0.9233 & 0.0005  \\
& & Sl. & 146.8 & 54.53 & \textbf{0.0202} & 0.0004 & 194.2 & 103.4 & 0.0483 & 0.0007 & 0.9165 & 0.0002\\ \\[-1.2ex]
& \multirow{5}{*}{50} & Gaus. &  0.0174 & 0.0006 & 0.0152 & 0.0005 & 0.0208 & 0.0006 & \textbf{0.0144} & 0.0005 & 0.9451 & 0.0005  \\
& & $t_3$  &  0.0180 & 0.0006 & \textbf{0.0154} & 0.0005 & 0.0217 &  0.0006 & 0.0157 & 0.0005 & 0.9450 & 0.0005\\ 
& & $st_{3, 0.5}$  & 0.0261  & 0.0008 & \textbf{0.0226} & 0.0006 & 0.0284 & 0.0008 &   0.0237 & 0.0007 & 0.9527 & 0.0007 \\
& &  M. Gaus. &  0.0193 & 0.0006 & \textbf{0.0155}	 & 0.0004 &  0.0243 & 0.0007 &  0.0214 &0.0004 & 0.9453 & 0.0005 \\
& &  Sl. &  1828.5 & 1402.1 & \textbf{0.0133} & 0.0004 & 492.4 & 257.5 & 0.0823 & 0.0008 & 0.9403 & 0.0002\\ \\[-1.2ex] 
\multirow{10}{*}{$0.5$} & \multirow{5}{*}{20} & Gaus.  & 0.0202 &  0.0005 & 0.0203 & 0.0004 &  0.206 & 0.0014 & \textbf{0.0175} &  0.0004 & 0.9219 & 0.0004 \\
& & $t_3$ & 0.0264  & 0.0007 & \textbf{0.0238} & 0.0005 & 0.213 &  0.0016 &  0.0241 & 0.0005 & 0.9240 & 0.0006 \\
& & $st_{3, 0.5}$  &  0.0879  & 0.0018 & \textbf{0.0725} & 0.0014 & 0.2206 & 0.0013 &   0.0814 & 0.0016 & 0.9784 & 0.0016 \\
& & M. Gaus. & 0.0501 & 0.0008 &  \textbf{0.0256} & 0.0005 & 0.2613 &   0.0032 &  0.0495 & 0.0007 & 0.9249 & 0.0005 \\
& & Sl. &  429.26 & 139.51 &  \textbf{0.0342} & 0.0005 & 199.1 &  82.75 & 0.1778 & 0.0022 & 0.9204 & 0.0003 \\ \\[-1.2ex]
& \multirow{5}{*}{50} & Gaus. &  0.0185 & 0.0006 & \textbf{0.0164} & 0.0005 & 0.0224 & 0.0006 & 0.0181 & 0.0005 & 0.9454 &0.0005 \\
& & $t_3$  &  0.0205 & 0.0006 & \textbf{0.0176} & 0.0005 & 0.0260 & 0.0006 & 0.0249 & 0.0005 &  0.9461 & 0.0005\\
& & $st_{3, 0.5}$  & 0.0808  & 0.0017 & \textbf{0.0660} & 0.0014 & 0.0835 & 0.0017 &  0.0833 & 0.0016 & 0.9976  & 0.0015\\
& &  M. Gaus. & 0.0312 & 0.0006  & \textbf{0.0188}	 & 0.0005 &  0.0418 & 0.0007 & 0.0598 & 0.0006 & 0.9462 &0.0005 \\
& &  Sl. &  6807.5 & 6111.1 & \textbf{0.0208} & 0.0004 &  2259.5 & 2036.4 & 0.3813 & 0.0034 &  0.9408 & 0.0002 \\ \\[-1.2ex] 
\multirow{10}{*}{$1$} & \multirow{5}{*}{20} & Gaus.  &  0.0283 &  0.0006 &   0.0312 & 0.0005 &  0.2202 & 0.0018 & \textbf{0.0306} & 0.0005 & 0.9252 & 0.0005\\
& & $t_3$ & 0.0481  & 0.0010& \textbf{0.0385} & 0.0006&  0.2556 & 0.0030 &   0.0525 &  0.0007 & 0.9288 & 0.0006 \\
& & $st_{3, 0.5}$  & 0.2968  & 0.0038 & \textbf{0.2168} & 0.0028 & 0.3929 & 0.0036 &  0.2817 & 0.0034 & 1.1229 & 0.0031 \\
& & M. Gaus.&  0.1375 &  0.0018 &  \textbf{0.0411} & 0.0006 & 0.3872 &  0.0054 & 0.1381 & 0.0016 &0.9308 & 0.0006 \\
& & Sl. & 6758237 & 6752967 & \textbf{0.0766} & 0.0010 &  1670998 & 1665878 &0.6233 & 0.0074 &  0.9390& 0.0009\\ \\[-1.2ex]
& \multirow{5}{*}{50} & Gaus. &   0.0179 & 0.0005 & \textbf{0.0175} & 0.0004 & 0.0242 & 0.0005 & 0.0288 & 0.0004 & 0.9387 &  0.0002\\
& & $t_3$  &  0.0295 & 0.0006 & \textbf{0.0236} & 0.0005 & 0.0399 & 0.0007 & 0.0545 & 0.0005 & 0.9476 & 0.0005\\
& & $st_{3, 0.5}$  & 0.2749  & 0.0036 & \textbf{0.1963} & 0.0027 & 0.2809 & 0.0036 &  0.2853 &0.0032 & 1.132&  0.0028\\
& &  M. Gaus. &  0.0739 & 0.0010 & \textbf{0.0249}	 & 0.0006 &   0.0947 & 0.0013 & 0.1877 & 0.0013 & 0.9483 & 0.0006\\
& &  Sl. &  99036 & 91566 & \textbf{0.0391} & 0.0005 & 4049.4 & 2045.1 & 1.4568 &  0.0129 & 0.9453 &  0.0003\\
\end{tabular}}
\caption{Mean and standard error of the MSE for the competing estimators over 1000 datasets of size $n=60$ with mean function $\mu(t)=\sin(6 \pi t)(t+1)$. Best performances are in bold.}
\label{Tab:ROSM1}
\end{table}

To evaluate the performance of the estimators we calculate their mean-square error (MSE), given by
\begin{equation*}
\MSE = m^{-1} \sum_{j=1}^m | \widehat{\mu}(T_j) - \mu(T_j) |^2,
\end{equation*}
which for large grids is an approximation to the $\mathcal{L}^2([0,1])$ distance.  Tables \ref{Tab:ROSM1} and \ref{Tab:ROSM2} below report the mean-squared errors and their standard errors based on 1000 simulated datasets of size $n=60$. In our experience such a moderately small sample size occurs fairly often in practice, see e.g., the well-known Canadian weather dataset \citep{Ramsay:2005}.

\begin{table}[H]
\centering
\resizebox{\columnwidth}{!}{%
\begin{tabular}{ccccccccccccc}
& \multicolumn{1}{c}{} &  \multicolumn{1}{c}{}  & \multicolumn{2}{c}{$\widehat{\mu}_{LSSP}$} & \multicolumn{2}{c}{$\widehat{\mu}_{HSP}$} & \multicolumn{2}{c}{$\widehat{\mu}_{LSLP}$} & \multicolumn{2}{c}{$\widehat{\mu}_{HF}$}  & \multicolumn{2}{c}{$\widehat{\mu}_{RS}$} \\ \\[-2ex]
$\sigma$ & $ m$ & Dist. & Mean & SE & Mean & SE & Mean & SE & Mean & SE  & Mean & SE   \\ \\
\multirow{10}{*}{$0.2$} & \multirow{5}{*}{20} & Gaus.  & 0.0115 &   0.0005 & \textbf{0.0091} & 0.0003 &  0.0134 & 0.0006 & 0.0127 & 0.0004 &  0.0133 &  0.0001 \\
& & $t_3$ & 0.0173  & 0.0006 & \textbf{0.0151} & 0.0005 & 0.0175 & 0.0006 &  0.0158 & 0.0005 & 0.0225 & 0.0005  \\
& & $st_{3, 0.5}$  & 0.0275 & 0.0009 & \textbf{0.0239} & 0.0007 & 0.0264 & 0.0009 &  0.0250 &  0.0007 & 0.0311 & 0.0007 \\
& & M. Gaus. &  0.0182 &   0.0005 &  \textbf{0.0151} & 0.0004 & 0.0206 &   0.0005 &  0.0204 & 0.0005 & 0.0222 & 0.0004  \\
& & Sl. &  132.50 & 76.288 & \textbf{0.0122} & 0.0003 & 202.09 & 103.15 & 0.0462 & 0.0007 & 0.0150 & 0.0002 \\ \\[-1.2ex]
& \multirow{5}{*}{50} & Gaus. &   0.0099 & 0.0004 & \textbf{0.0080} & 0.0003 & 0.0111 & 0.0004 & 0.0117 & 0.0004 & 0.0137 & 0.0001\\
& & $t_3$  &  0.0156 & 0.0005 & \textbf{0.0138} & 0.0004 & 0.0158 & 0.0005 & 0.0154 & 0.0005 & 0.0239 & 0.0004 \\ 
& & $st_{3, 0.5}$  & 0.0262  & 0.0009 &  \textbf{0.0232} & 0.0008 & 0.0258 & 0.0009 &  0.0255 & 0.0008 & 0.0330&  0.0008\\
& &  M. Gaus. & 0.0168 & 0.0006 & \textbf{0.0137} &0.0004 &   0.0177 & 0.0006 & 0.0216 & 0.0005 & 0.0238 & 0.0004 \\
& &  Sl. &  564.14 & 394.85 & \textbf{0.0097} & 0.0004 & 359.49 & 233.32 & 0.0831 & 0.0009 &  0.0150 & 0.0001\\ \\[-1.2ex] 
\multirow{10}{*}{$0.5$} & \multirow{5}{*}{20} & Gaus.  & 0.0121 &   0.0004 &  \textbf{0.0102} &  0.0003 &  0.0149 & 0.0004 & 0.0163 &  0.0004 & 0.0141 & 0.0001 \\
& & $t_3$ & 0.0192  & 0.0005 & \textbf{0.0168} & 0.0005 & 0.0225 &  0.0006 &  0.0237 & 0.0005 & 0.0235 &  0.0005 \\
& & $st_{3, 0.5}$  &  0.0817 & 0.0020 & \textbf{0.0669} & 0.0016 & 0.0794 & 0.0020 &   0.0818 & 0.0017 &  0.0728 & 0.0015 \\
& & M. Gaus. & 0.0313 &  0.0007 &  \textbf{0.0182} & 0.0005 & 0.0435 &   0.0010 &   0.0492 & 0.0007 & 0.0242 & 0.0005 \\
& & Sl. &   13135 & 11580 &   \textbf{0.0175} & 0.0002 & 9789.1 &  8069.6 & 0.1827 & 0.0021 & 0.0181 & 0.0002 \\ \\[-1.2ex]
& \multirow{5}{*}{50} & Gaus. &  0.0105 &  0.0005 & \textbf{0.0084} & 0.0003 & 0.0119 & 0.0005 & 0.0155 & 0.0004 & 0.0142 & 0.0001 \\
& & $t_3$  &  0.0174 & 0.0006 &  \textbf{0.0146} & 0.0005 & 0.0185& 0.0006 & 0.0243 & 0.0005 & 0.0245 & 0.0005 \\
& & $st_{3, 0.5}$  & 0.0795  & 0.0019 & \textbf{0.0634} & 0.0015 &  0.0775 & 0.0018 & 0.0832 & 0.0016 &  0.0728 & 0.0015 \\
& &  M. Gaus. & 0.0218 &  0.0006 &  \textbf{0.0154}	 & 0.0005&  0.0256 & 0.0006 &  0.0587 & 0.0006 & 0.0248 & 0.0005 \\
& &  Sl. &   295.43 &  125.48 & \textbf{0.0069} & 0.0001 & 233.94 & 96.796 & 0.3917 & 0.0035 & 0.0153 & 0.0001 \\ \\[-1.2ex] 
\multirow{10}{*}{$1$} & \multirow{5}{*}{20} & Gaus.  &  \textbf{0.0151} &   0.0004 & 0.0152 & 0.0004 &  0.0214 & 0.0005 &  0.0290 & 0.0005 & 0.0179 & 0.0003  \\
& & $t_3$ & 0.0293  & 0.0007 & \textbf{0.0244} & 0.0006 &  0.0412 &  0.0009 &  0.0516 &0.0007 & 0.02848 &  0.0006 \\
& & $st_{3, 0.5}$  & 0.2702 &  0.0035 & \textbf{0.1957} & 0.0027 & 0.2664& 0.0036 &  0.2704 & 0.0031 &   0.1991&  0.0027\\
& & M. Gaus.&  0.0727 & 0.0015 & \textbf{0.0291} & 0.0007 & 0.1155 & 0.0028 & 0.1406 & 0.0016 &0.03097 &  0.0007\\
& & Sl. & 2082.9 & 892.89 & 0.0399 & 0.0006 &  2061.4 & 848.98 & 0.6225 & 0.007 &  \textbf{0.0341} & 0.0006 \\ \\[-1.2ex]
& \multirow{5}{*}{50} & Gaus. &  0.0115  & 0.0004 & \textbf{0.0102} &  0.0003 & 0.0134 & 0.0004 & 0.0279 & 0.0004 & 0.0157 & 0.0001 \\
& & $t_3$  &  0.0213 & 0.0006 & \textbf{0.0175} & 0.0005 & 0.0246 & 0.0006 & 0.0552 &  0.0005 &  0.0261 &  0.0005 \\
& & $st_{3, 0.5}$  & 0.2594  & 0.0035 & \textbf{0.1852} & 0.0027 & 0.2527& 0.0034 &  0.2794 & 0.0032 & 0.1935 & 0.0027\\
& &  M. Gaus. &  0.0391 & 0.0008 & \textbf{0.0183}	 & 0.0005 &  0.0483 & 0.0009 & 0.1901 & 0.0014 & 0.0267 & 0.0005 \\
& &  Sl. &  4171.4 &  3138.2& \textbf{0.0157} & 0.0002 & 10696.9 & 9682.4 & 1.4655 &  0.0131 &  0.0216 &  0.0002 \\
\end{tabular}}
\caption{Mean and standard error of the MSE for the competing estimators over 1000 datasets of size $n=60$ with mean function $\mu(t)= 3\exp[-(0.25-t)^2/0.1]$.}
\label{Tab:ROSM2}
\end{table}

The results indicate that the least-squares smoothing spline estimator, $\widehat{\mu}_{LSSP}$, behaves well under all noise and discretization settings provided that the measurement errors follow a Gaussian distribution. However, its performance quickly deteriorates as soon as the errors deviate from the Gaussian ideal. Similar remarks apply for the local linear estimator $\widehat{\mu}_{LSLP}$, except that its behavior is not good for the sinusoidal mean function when $m$ is small. The reason for this lesser performance is that $\widehat{\mu}_{LSLP}$ regularly oversmooths in this case and thus misses the peaks and troughs of the sinusoidal mean function. We believe that its performance may be greatly improved if more effort is invested in the selection of its bandwidth, which naturally would come at the cost of a higher computing time.

Comparing the robust estimators $\widehat{\mu}_{HSP}$, $\widehat{\mu}_{HF}$ and $\widehat{\mu}_{RS}$ in detail leads to a number of interesting observations. First, while $\widehat{\mu}_{HF}$   performs well in general for lower levels of noise, it is highly sensitive to larger levels of noise, especially when that is accompanied by heavy-tailed measurement errors. In particular, for the Gaussian mixture  $\widehat{\mu}_{HF}$ tends to perform nearly as bad as the least-squares  estimator $\widehat{\mu}_{LSSP}$ with respect to both mean functions. The regression spline estimator $\widehat{\mu}_{RS}$ does not perform well with respect to the sinusoidal mean function of our first example. This may be attributed to the local characteristics of this function. For such functions regression spline estimators require not only an appropriate number of knots but also good placement, see, e.g., \citep{Eubank:1999} for possible strategies in that respect. Inevitably, this results in much increased computational effort that negates the computational advantage of regression spline estimators. The performance of $\widehat{\mu}_{RS}$ greatly improves in our second example, but it gets outperformed by the competing estimators in all but one of the settings considered.

The Huber-type smoothing spline estimator, $\widehat{\mu}_{HSP}$, performs as well as $\widehat{\mu}_{LSSP}$ under Gaussian errors and exhibits a high degree of resilience towards high levels of noise and contamination. In particular, it can be seen that symmetric contamination, i.e. $t_3$, mixture Gaussian errors and Slash errors, only has a small effect on the performance of the estimator which clearly outperforms its competitors. Interestingly, $\widehat{\mu}_{HSP}$ as well as the other robust estimators seem more vulnerable to asymmetric contamination resulting from the right-skewed t-distribution. Although its performance deteriorates in this case, the Huber-type smoothing spline estimator best mitigates the effect of this contamination relative to its competitors. Overall, the present simulation experiments suggest that robust smoothing spline estimators perform well in clean data and safeguard against outlying observations either in the form of outlying curves or heavy-tailed measurement errors.

\section{Application: Covid-19 in Europe}
\label{sec:ROSM6}

Since its identification at the end of $2019$, the Covid-19 virus has been responsible for hundreds of millions of infected cases and considerable economic and social upheaval. In this study we aim to shed light on the course of the epidemic in Europe and identify the overall trends as well as countries which have been most or least afflicted. Our dataset consists of the number of daily new cases per million persons from $34$ European countries in the period between the 24th of January 2020 and the 21st of March 2021 for a total of $423$ days. The data is not observed on the same grid, as most countries only began systematically reporting the number of new daily cases around the end of March 2020.  Hence, according to our previous definitions, we are in the situation of an independent design. The data is depicted in the left panel of Figure \ref{fig:ROSM3}. 

Despite the fact that the raw data exhibits a number of isolated spikes corresponding to occasionally very large numbers of daily cases, two peaks are easily discernible. The first peak taking place between the end of March and the beginning of April is commonly identified with the culmination of the first wave of infections. The number of infections then plummeted in the aftermath of the lock-down measures, but rose rapidly after the end of the summer possibly due to the gradual abandonment of restrictions. The so-called second wave of infections peaked towards the end of 2020 and subsided afterwards. However, as of February 2021 the number of infections seems to be yet again on the increase. In order to draw formal conclusions, we employ the quantile cubic smoothing spline with values of $\tau \in \{0.1, 0.3, 0.5, 0.7, 0.9\}$ yielding the right panel of Figure~\ref{fig:ROSM3}.

\begin{figure}[H]
\centering
\subfloat{\includegraphics[width = 0.49\textwidth]{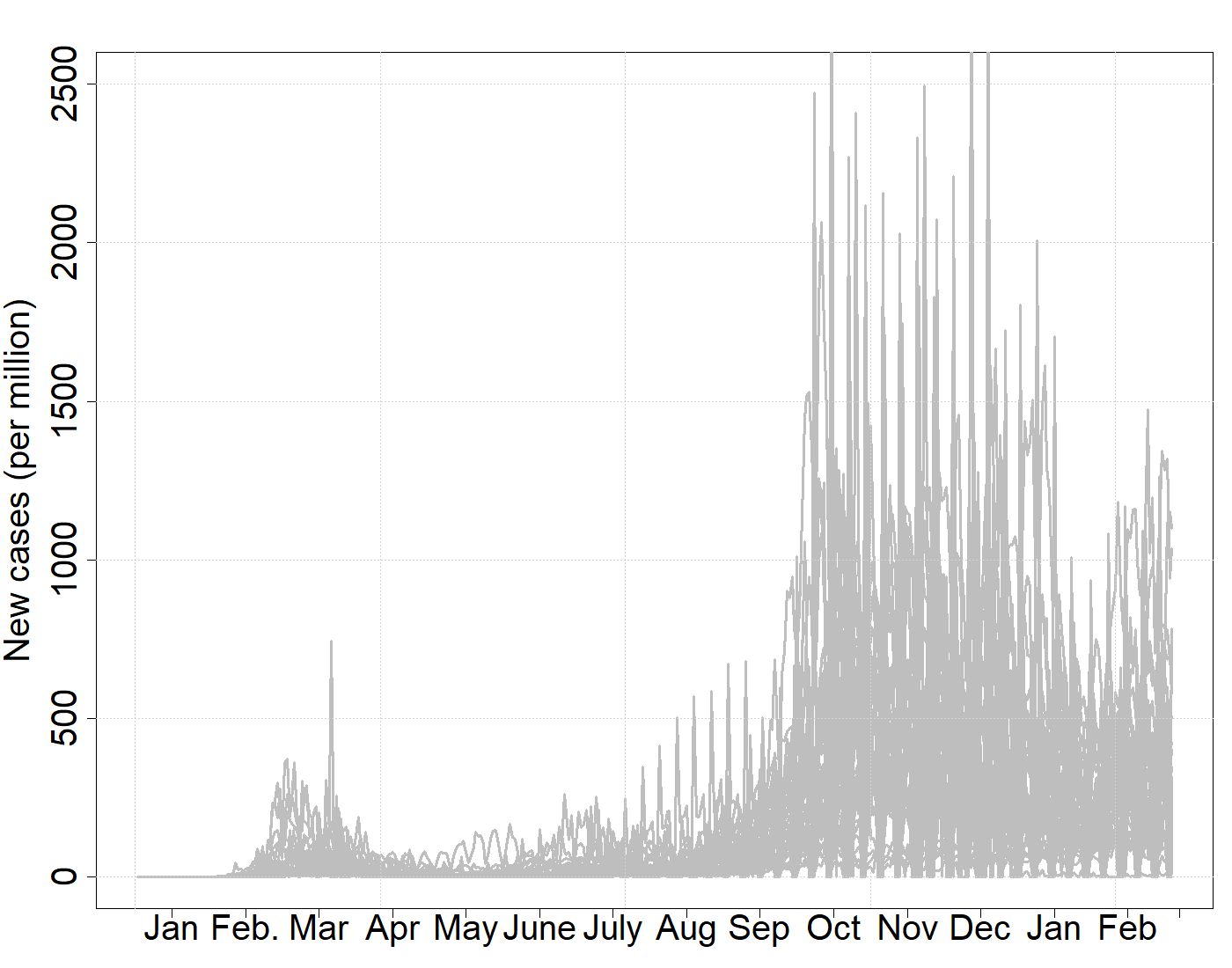}} \ 
\subfloat{\includegraphics[width = 0.49\textwidth]{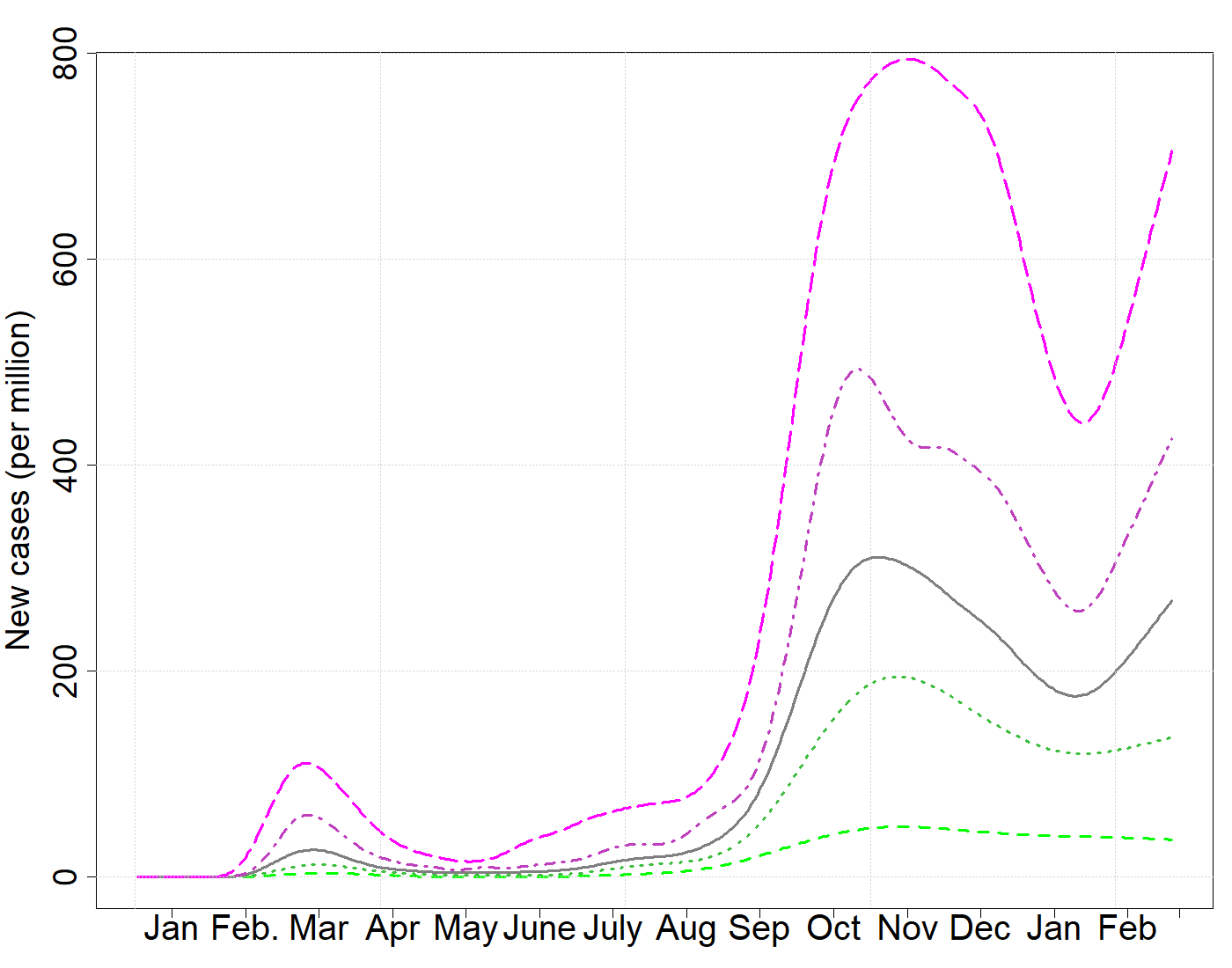}}
\caption{Left: number of daily new cases per million population in $34$ European countries. Right: estimated quantiles with ( \dashed, \dotted, \full,  \dotdash, \denselydashed ) corresponding to $(0.1, \ 0.3, \ 0.5, \ 0.7, \ 0.9)$-quantiles respectively. }
\label{fig:ROSM3}
\end{figure}

The estimated quantiles serve to confirm the suspected bimodality in the number of infections. Although these two peaks are identified by all estimated quantiles, higher quantiles significantly differ from lower ones in that they are much steeper. This suggests that countries most afflicted by the first wave of infections experienced only a brief period of calm in between the two waves. For these countries the number of new infections rose very fast towards the end of the summer. By contrast, the median curve indicates that for countries less affected by the first wave, i.e., countries with fewer than $25$ cases per million persons, the number of infections only starts to gradually increase again in August. The estimates further suggest a structural break in early September, when the number of new infections picks up pace. This change may be attributable to the opening of schools in many countries. 

To visualize the estimates, we have constructed heat maps of Europe based on the quantiles for four days covering the whole period. These are shown in Figure \ref{fig:ROSM4}.
The heat maps in the top row indicate that Spain, Portugal, Luxembourg Italy, Switzerland, Sweden and the United Kingdom were hit hardest in the first wave of the disease. Similarly, from the heat maps in the bottom row it can be seen that countries worse afflicted by the second wave include Luxembourg, Switzerland, Serbia, Hungary, Poland and Estonia. A tentative conclusion from these maps may be that the first wave was more detrimental to western European countries, while during the second wave it was the countries of eastern Europe that were worse afflicted.

\begin{figure}[H]
\centering
\subfloat{\includegraphics[width = 0.492\textwidth]{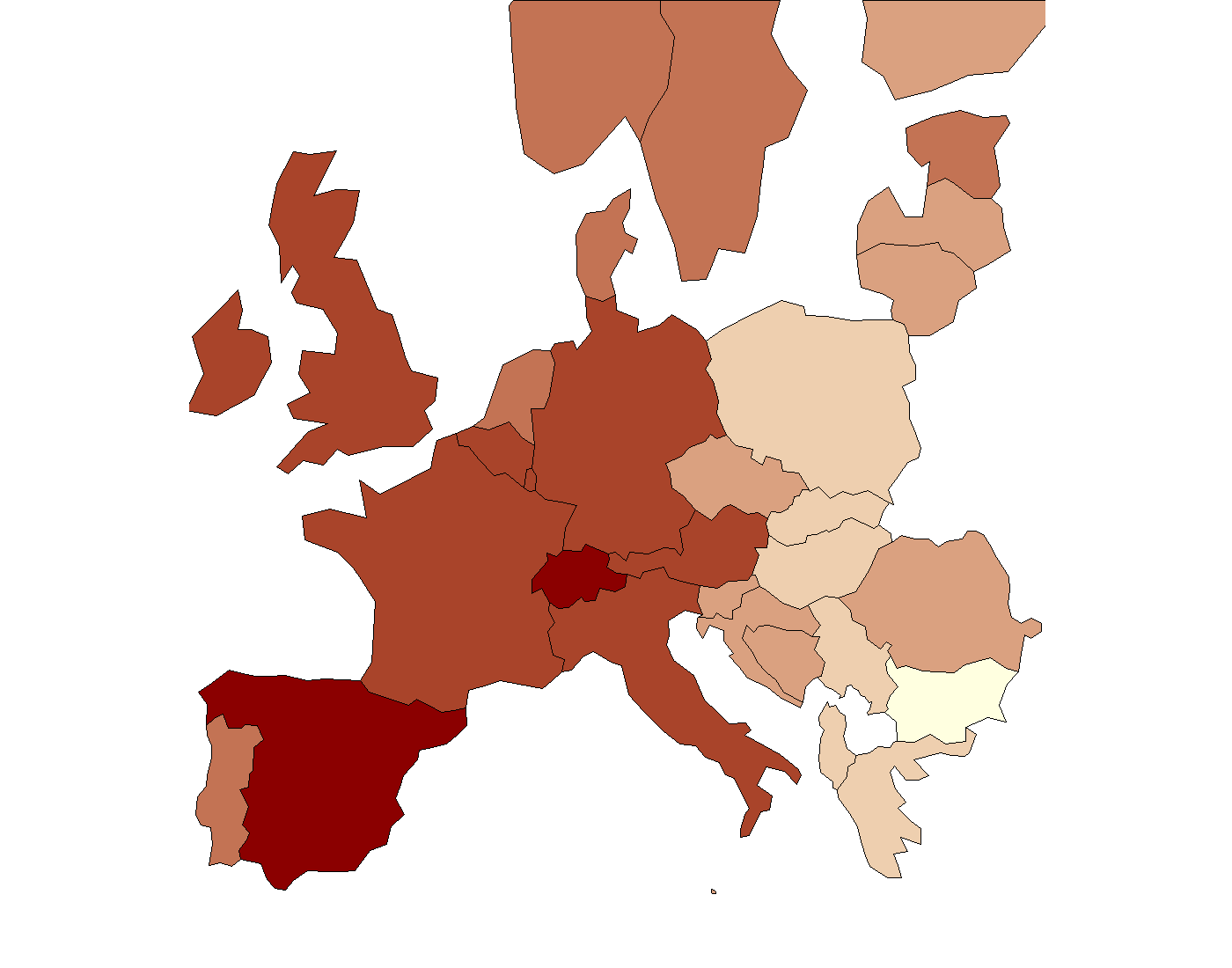}} \ 
\subfloat{\includegraphics[width = 0.492\textwidth]{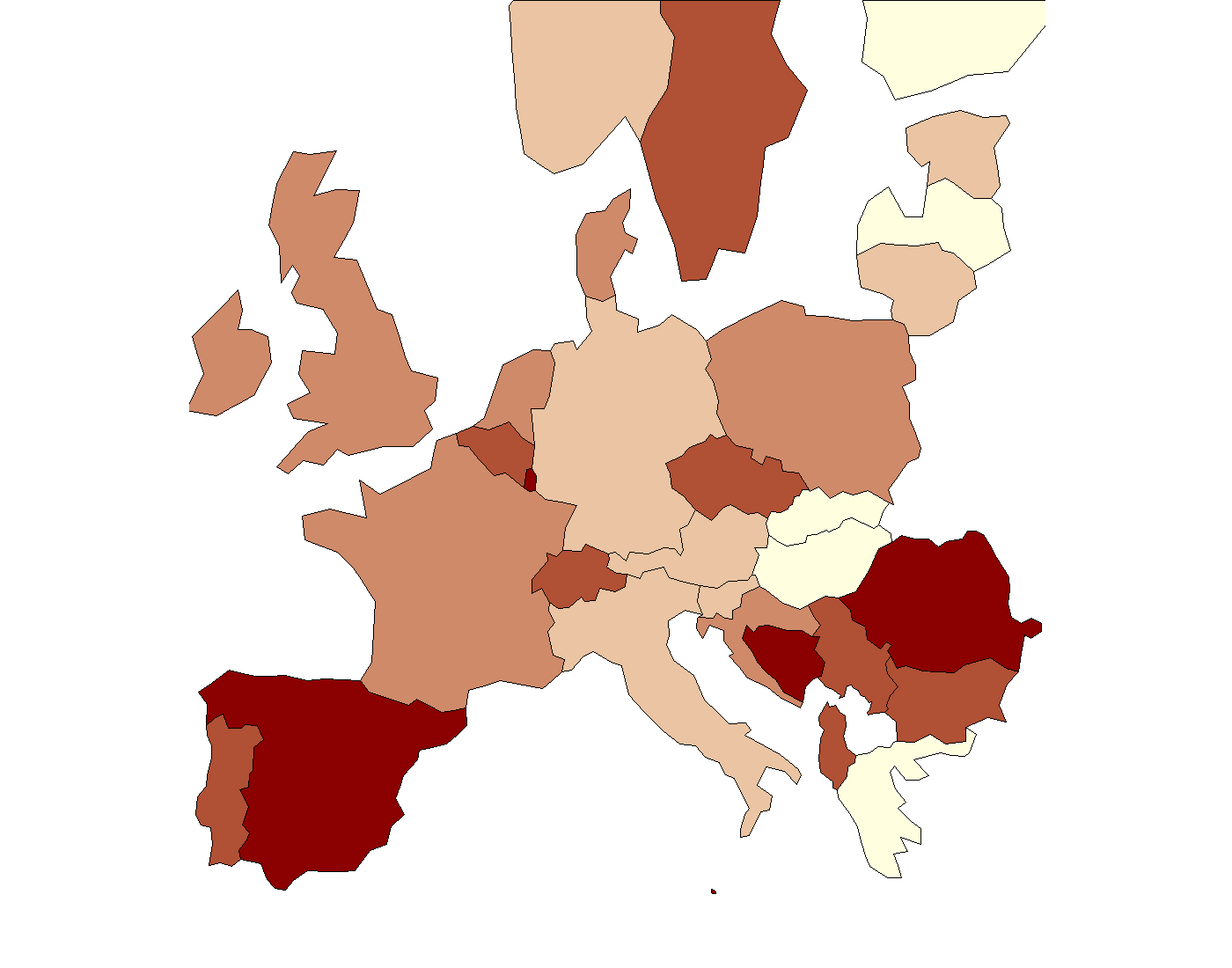}} \\
\subfloat{\includegraphics[width = 0.492\textwidth]{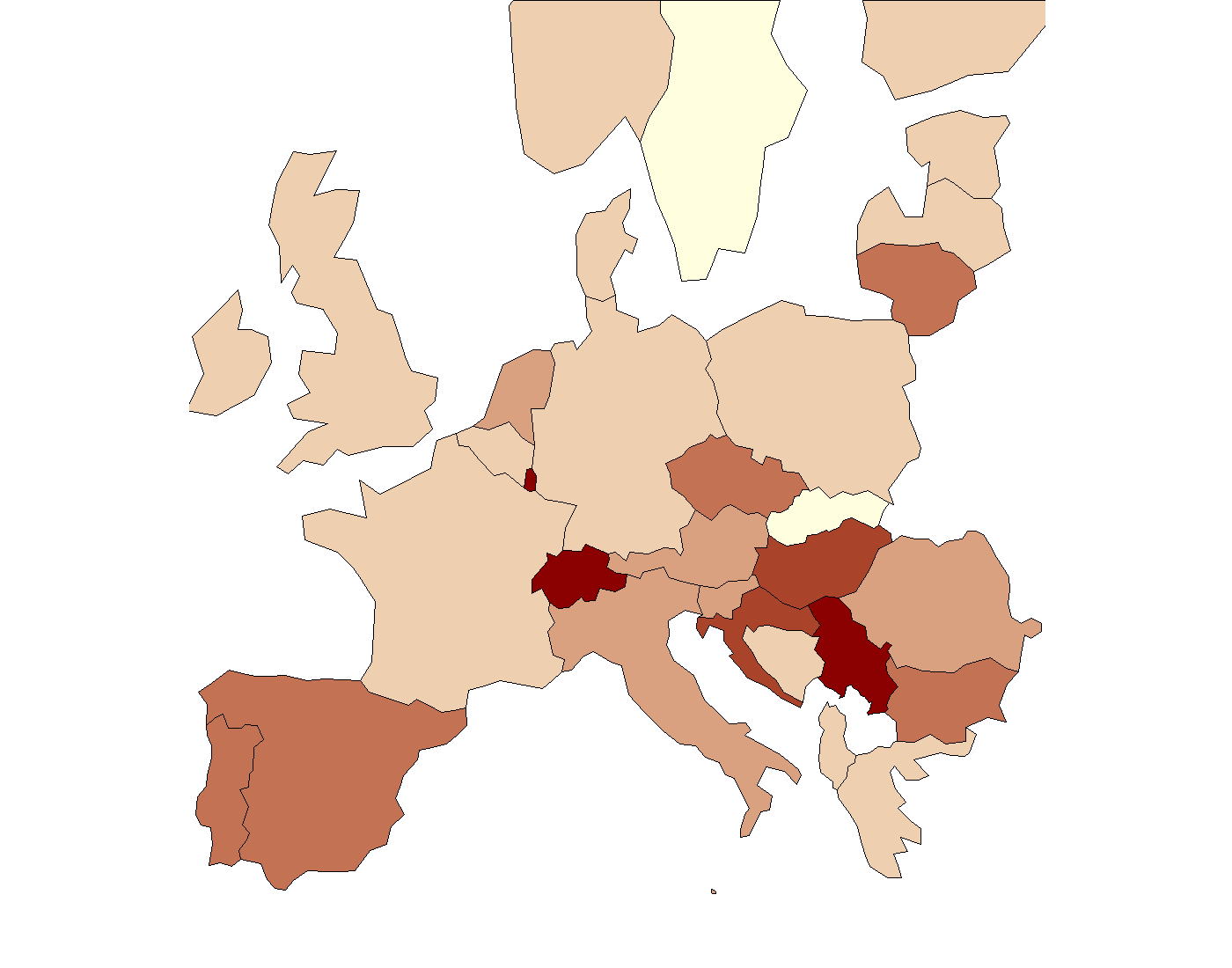}} \
\subfloat{\includegraphics[width = 0.492\textwidth]{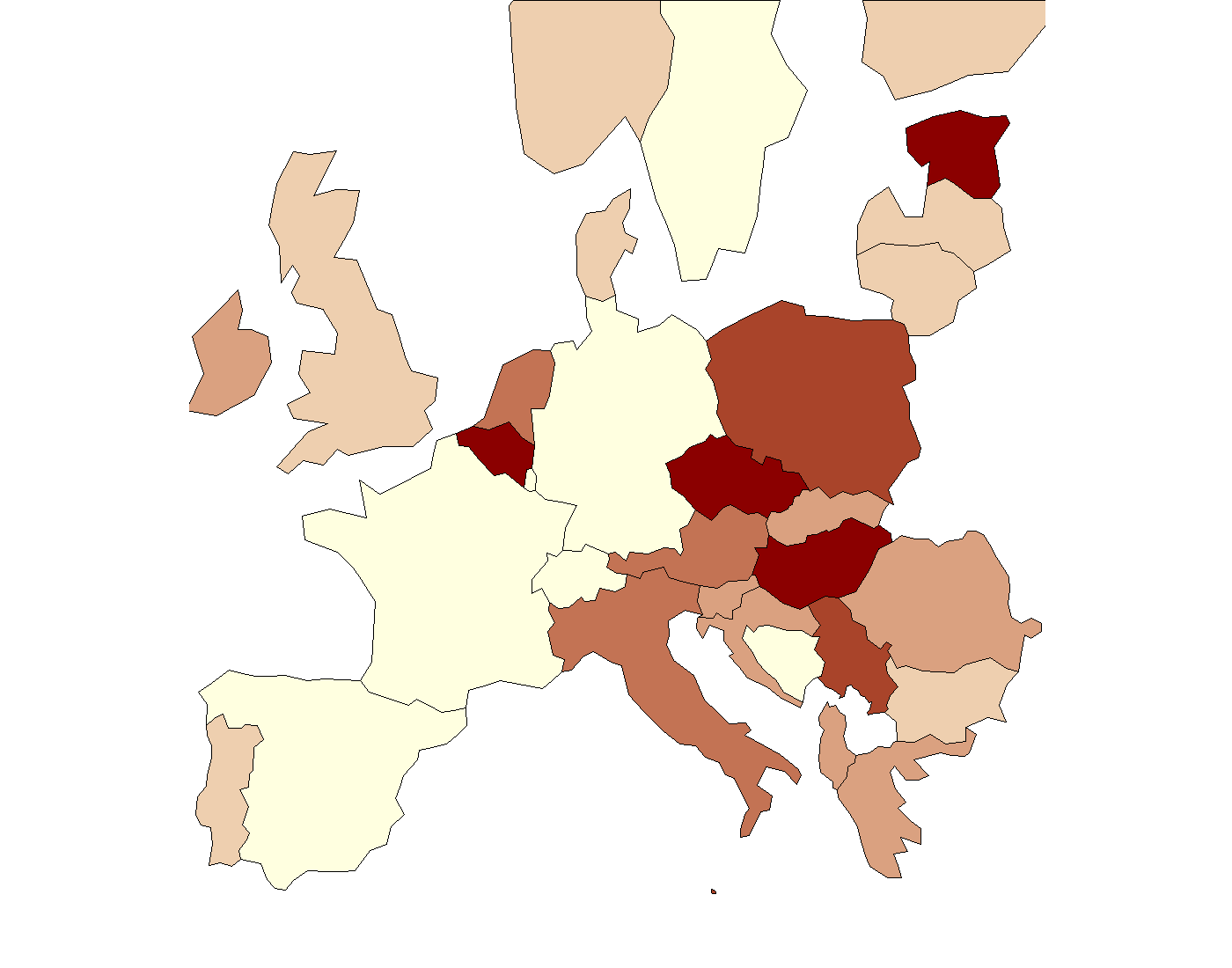}}
\caption{Heat maps of Europe on selected dates. Top row: March 30th and July 30th 2020. Bottom row: November 30th 2020 and March 21st 2021. Darker colors correspond to higher quantiles of the number of new infections.}
\label{fig:ROSM4}
\end{figure}
Lastly, the ab	ove analysis can be complemented by a study of the daily deaths per million persons attributed to Covid-19. The left panel of Figure~\ref{fig:ROSM5} plots the corresponding trajectories while the right panel presents the quantile smoothing spline estimates. It can clearly be seen that the number of daily deaths is also bimodal, but the second peak is only marginally larger than the first peak. In other words, the numbers of deaths have not kept up with the number of cases, which exploded during the second wave. Two plausible explanations for this difference is that the fraction of younger less vulnerable people that contract the disease is much higher than in the first wave of the epidemic and the treatment of patients has improved. 

\begin{figure}[H]
\centering
\subfloat{\includegraphics[width = 0.495\textwidth]{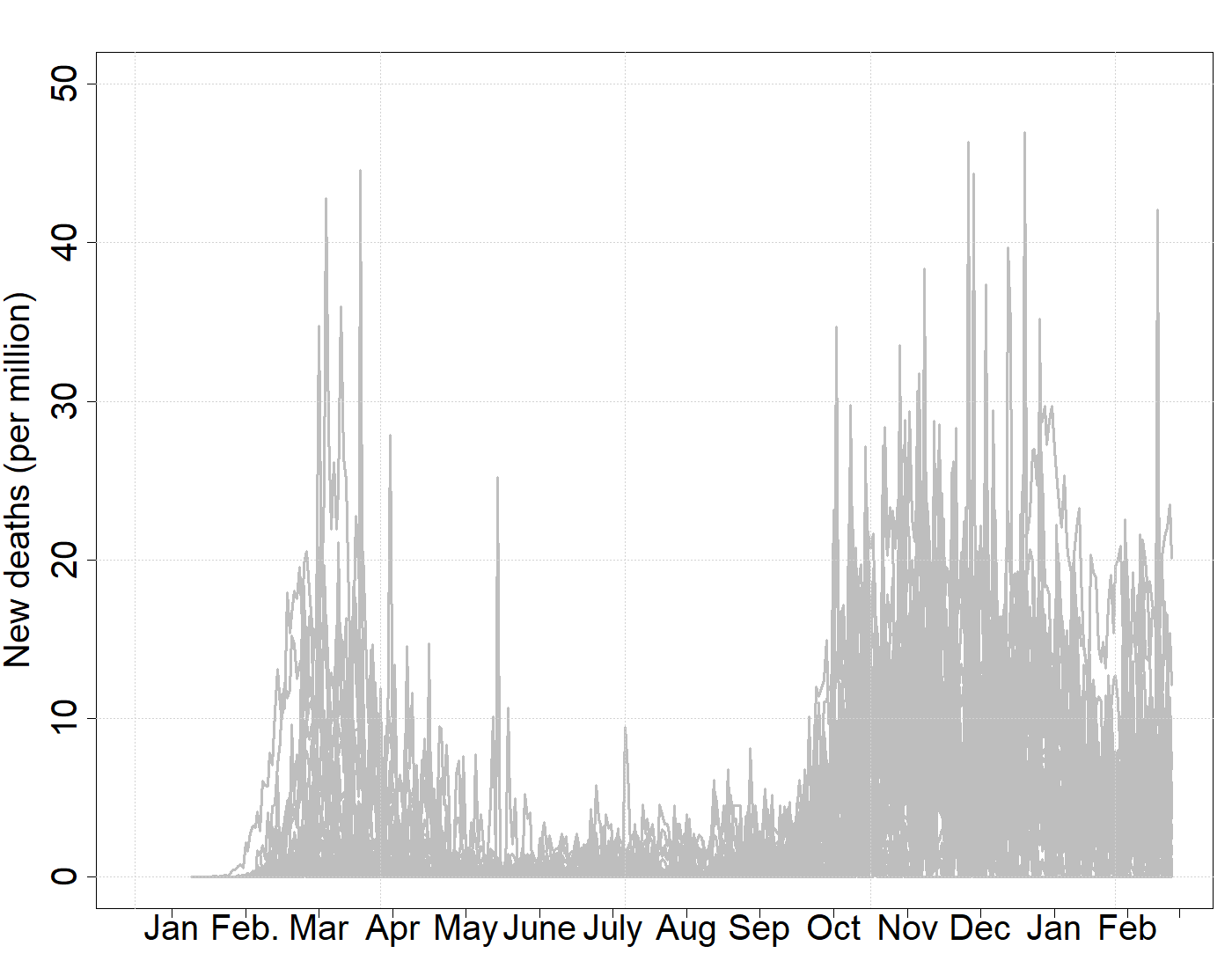}} \ 
\subfloat{\includegraphics[width = 0.495\textwidth]{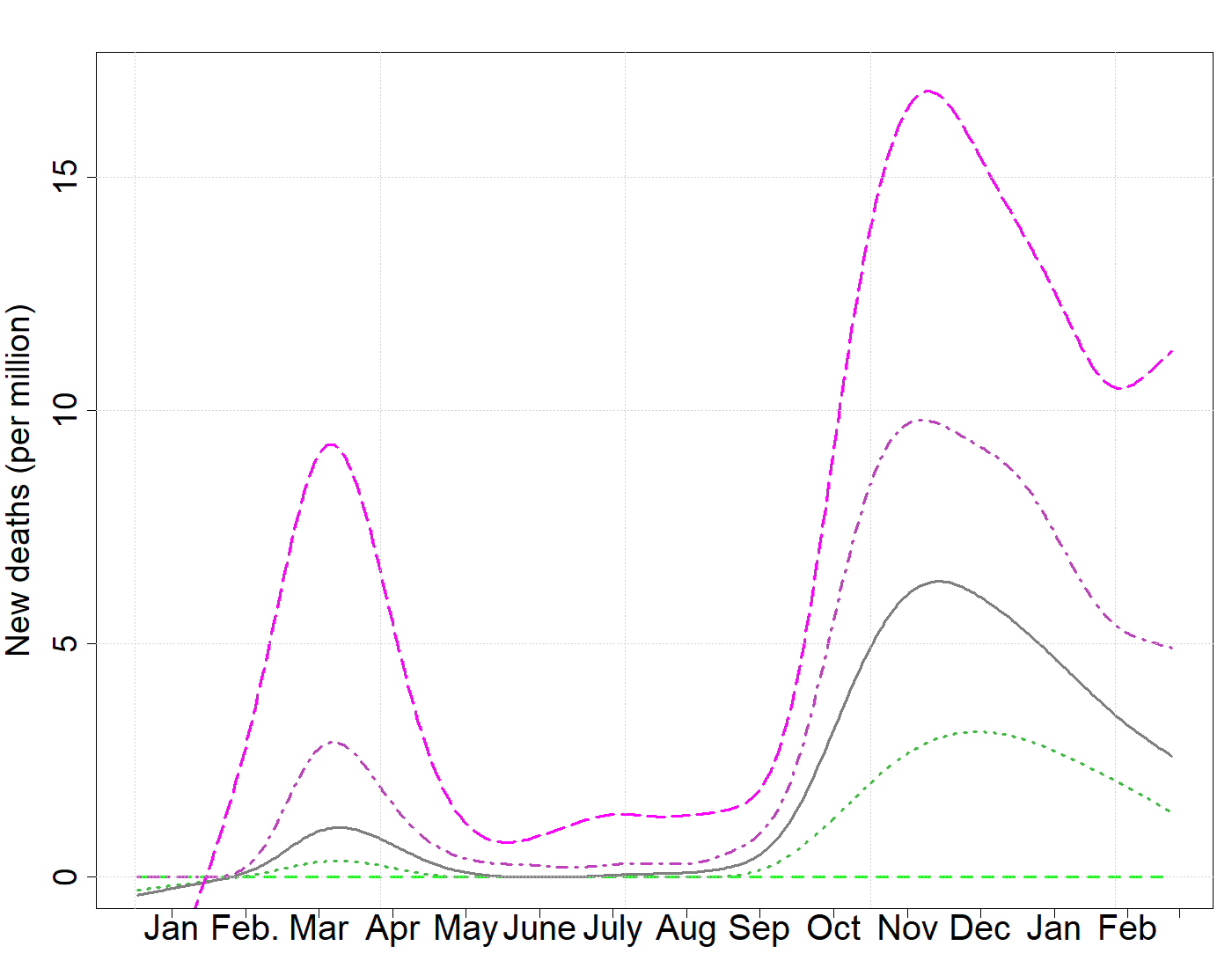}}
\caption{Left: number of daily new deaths per million population in $34$ European countries. Right: estimated quantiles with ( \dashed, \dotted, \full,  \dotdash, \denselydashed ) corresponding to $(0.1, \ 0.3, \ 0.5, \ 0.7, \ 0.9)$-quantiles respectively. }
\label{fig:ROSM5}
\end{figure}

\section{Concluding remarks}

This paper provides theoretical justification for robust smoothing spline estimators under both common and  independent designs with trajectories that may be densely or sparsely sampled. From here there are several possible research directions worth pursuing. Of particular importance is robust inference for the covariance structure of second-order processes. Despite its importance, this problem has not received much attention. To the best of our knowledge current work, such as \citep{Pan:2012}, require completely observed trajectories. To construct robust dispersion estimates from discretely sampled functional data the current framework can be extended by replacing the one-dimensional smoothing spline with a two-dimensional thin plate spline \citep{Wahba:1990, Wood:2017}. Robust estimation of the covariance structure would then allow for robust functional principal component analysis based on discretely sampled data in the manner outlined in \citep{Yao:2005, Li:2010}.

Another important area for the application of robust smoothing spline estimators which has attracted great interest recently is functional regression and its variants. To date, most estimation proposals assume that the functional predictor is fully observed. Hence, the discreteness of the data is essentially ignored. A smoothing spline estimator for discretely observed curves based on the MM principle was proposed by \citet{maronna2013robust}, 
but without any theoretical support. We are confident that under appropriate conditions this MM-type smoothing spline estimator can achieve optimal rates of convergence, as defined in \citep{Crambes:2009}, whilst also providing a much safer estimation method in practice. We aim to study this in future work.

\section{Appendix: Proofs of the theoretical results}
\label{sec:ROSM8}

\begin{proof}[Proof of Proposition 1]
First, we recall the existence of a complete orthonormal sequence $\{\phi_j\}_j$ for $\mathcal{L}^2([0,1])$ such that for $i, j \geq 1$,
\begin{align*}
\langle \phi_i^{(r)}, \phi_j^{(r)} \rangle_{2} = \gamma_i \delta_{ij},
\end{align*}
for values $0 = \gamma_1 = \ldots = \gamma_r < \gamma_{r+1} < \ldots$ which satisfy 
\begin{align*}
C_1 k^{2r}  \leq \gamma_{k+r} \leq C_2 k^{2r}, \quad k \geq 1,
\end{align*}
for some $0 < C_1 \leq C_2$, see \citep[Theorem 2.8.3]{Hsing:2015}. Furthermore, since $\mathcal{W}^{r,2}([0,1]) \subset \mathcal{L}^2([0,1])$ and the $\phi_j$ are orthogonal in $\mathcal{W}^{r,2}([0,1])$ under $\langle \cdot, \cdot \rangle_{r, \lambda}$, we have that
\begin{align*}
\langle \phi_i, \phi_j \rangle_{r, \lambda} = (1 + \lambda \gamma_i) \delta_{ij},
\end{align*}
such that $\phi_j/(1+\lambda \gamma_j)$ is an orthonormal basis in $\mathcal{W}^{r,2}([0,1])$ under the inner product $||\cdot||_{r, \lambda}$. It follows that for every $f \in \mathcal{W}^{r,2}([0,1])$ we can write
\begin{align*}
f(x) = \sum_{j=1}^{\infty} \frac{f_j}{(1+\lambda \gamma_j)^{1/2}} \phi_j(x),
\end{align*}
for some square summable sequence $\{f_j\}_j$. Using this representation and the Schwarz inequality we find
\begin{align*}
|f(x)| \leq \sup_{j \geq 1} \sup_{x \in [0,1]} |\phi_j(x)| \left\{ \sum_{j=1}^{\infty} |f_j|^2 \right\}^{1/2} \left\{ \sum_{j=1}^{\infty} \frac{1}{1+\lambda \gamma_j} \right\}^{1/2}.
\end{align*}
The first factor on the right hand-side is bounded, say $\sup_{j \geq 1} \sup_{x \in [0,1]} |\phi_j(x)| \leq M$, as the $\phi_j$ are uniformly bounded in both $x$ and $j$, see \citep[p. 60]{Hsing:2015}. Moreover, Parseval's theorem yields
\begin{align*}
\sum_{j=1}^{\infty} |f_j|^2 = ||f||_{r, \lambda}^2 < \infty,
\end{align*}
since $f \in \mathcal{W}^{r,2}([0,1])$. Finally, by the properties of the $\gamma_j$ and integral approximation, we obtain
\begin{align*}
\sum_{j=1}^{\infty}  \frac{1}{1+\lambda \gamma_j}  \leq r + \int_{r}^{\infty} \frac{dx}{1+\lambda C_1 x^{2r}} \leq r + c_r \lambda^{-1/2r},
\end{align*}
for some finite constant $c_r$ depending only on $r$. Putting everything together and recalling that $\lambda \in (0,1]$,
\begin{align*}
|f(x)| \leq  M \max\{r, c_r\}^{1/2} \lambda^{-1/4r} ||f||_{r, \lambda},
\end{align*}
which establishes part (a). Part (b) now follows from the Riesz representation theorem for bounded linear functionals on $\mathcal{W}^{r,2}([0,1])$. To prove part (c) note that, since $x \mapsto \mathcal{R}_{r, \lambda}(x,y) \in \mathcal{W}^{r,2}([0,1])$ and $\phi_j/(1+\lambda \gamma_j)^{1/2}, j \geq 1$, is an orthonormal basis, we may write
\begin{align*}
\mathcal{R}_{r, \lambda}(x,y) = \sum_{j = 1}^{\infty} \frac{1}{1+\lambda \gamma_j} \langle \mathcal{R}_{r, \lambda}(\cdot,y), \phi_j \rangle_{r, \lambda} \phi_j(x).
\end{align*}
But $\mathcal{R}_{r, \lambda}(x,y)$ is the reproducing kernel, hence $\langle \mathcal{R}_{r, \lambda}(\cdot,y), \phi_j \rangle_{r, \lambda} = \phi_j(y)$ resulting in
\begin{align*}
\mathcal{R}_{r, \lambda}(x,y) = \sum_{j = 1}^{\infty} \frac{1}{1+\lambda \gamma_j} \phi_j(y) \phi_j(x),
\end{align*}
as asserted. Now, for every $n>m$ we have
\begin{align*}
\sum_{j=m+1}^n \frac{|\phi_j(x) \phi_j(y)|}{1+\lambda \gamma_j} \leq M^2 \sum_{j=m+1}^n \frac{1}{1+\lambda \gamma_j}.
\end{align*}
Since for every $\lambda>0$ the latter series is convergent, it follows that for every $\epsilon>0$ there exists an $N$ such that
\begin{align*}
\sup_{n,m \geq N} \sup_{x, y \in [0,1]^2}\sum_{j=m+1}^n \frac{|\phi_j(x) \phi_j(y)|}{1+\lambda \gamma_j} < \epsilon. 
\end{align*}
The absolute series is thus uniformly Cauchy and hence uniformly convergent, as was to be proven.
\end{proof}

Proposition~\ref{Prop:ROSM1} implies that
\begin{equation*}
||\mathcal{R}_{r, \lambda}(x, \cdot)||_{r, \lambda}^2 =  \langle \mathcal{R}_{r, \lambda}(x, \cdot), \mathcal{R}_{r, \lambda}(\cdot, x) \rangle_{r,\lambda} =  \mathcal{R}_{r, \lambda}(x, x)   \leq c_r  \lambda^{-1/{4r}}  ||\mathcal{R}_{r, \lambda}(x, \cdot)||_{r, \lambda}
\end{equation*}
for $\lambda \in (0,1]$. Upon dividing both sides with $||\mathcal{R}_{r, \lambda}(x, \cdot)||_{r, \lambda}$ we obtain
\begin{equation}
\label{eq:ROSM8}
\sup_{x \in [0,1]} ||\mathcal{R}_{r, \lambda}(x, \cdot)||_{r, \lambda} \leq c_r \lambda^{-1/{4r}}.
\end{equation}
We will make frequent use of this inequality in our proofs below.

First, we give two lemmas that will be used to proof Theorem~\ref{Thm:ROSM1}~.
The first lemma concerns the minimization of convex semi-continuous functionals in the Hilbert space and is proven in \citep{Kalogridis:2020}.
\begin{lemma}
\label{Lem:ROSM1}
Let $\mathcal{H}$ denote a real Hilbert space of functions endowed with inner product $\langle \cdot, \cdot \rangle_{\mathcal{H}}$ and associated norm $||\cdot||_{\mathcal{H}}$ and let $L: \mathcal{H} \to \mathbb{R}_{+}$ denote a convex lower semicontinuous functional. If
\begin{equation*}
L(0) < \inf_{||f||_{\mathcal{H}}=1 } L(f),
\end{equation*}
then there exists a minimizer of $L$ in the unit ball $\{f \in \mathcal{H}: ||f||_{\mathcal{H}} \leq 1 \}$
\end{lemma}
The second lemma concerns the approximation of an integral of a smooth $\mathcal{W}^{r,2}([0,1])$-function by a sum.

\begin{lemma}
\label{Lem:ROSM2}
Under assumption (A6) there exists a constant $c_0$ such that, for all $f \in \mathcal{W}^{r,2}([0,1])$ and $m \geq 2$,
\begin{equation*}
\left|\frac{1}{m} \sum_{j=1}^m |f(T_j)|^2 - \int_{0}^1 |f(t)|^2 dt \right| \leq \frac{c_0}{m \lambda^{1/{2r}}} ||f||_{r, \lambda}^2
\end{equation*}
\end{lemma}

\begin{proof}
We first claim the existence of a universal constant $c$ such that for all $f \in \mathcal{W}^{1,1}([0,1])$ 
\begin{align*}
\left|\frac{1}{m} \sum_{j=1}^m |f(T_j)| - \int_{0}^1 |f(t)| dt \right| \leq \frac{c_0}{m} \int_{0}^1 |f^{\prime}(t)|dt.
\end{align*}
Let $Q_m$ denote the distribution function that jumps $m^{-1}$ at each design point $T_j$. Integration by parts then shows
\begin{align*}
\left|\frac{1}{m} \sum_{j=1}^m |f(T_j)| - \int_{0}^1 |f(t)| dt \right| & = \left|\int_{0}^1 (Q_m(t)-t) f^{\prime}(t) dt \right| 
\\ & \leq \sup_{t \in [0,1]} |Q_m(t)-t| \int_{0}^1  |f^{\prime}(t)| dt.
\end{align*}
But,
\begin{align*}
\sup_{t \in [0,1]} |Q_m(t)-t| = \max_{j = 1, \ldots, m-1} \sup_{t \in [T_j, T_{j+1})}  |Q_m(t)-t| \leq c_0 m^{-1},
\end{align*}
by assumption (A6). The claim thus holds. The proof is completed with an appeal to \citep[Chapter 13, Lemma 2.27]{Eg:2009}.
\end{proof}

For notational convenience, in  the proof of Theorem~\ref{Thm:ROSM1} we denote positive constants by $c_0$. Note that this constant may change value from appearance to appearance.
\begin{proof}[Proof of Theorem~\ref{Thm:ROSM1}]
Let $L_n(f)$ denote the objective function, that is,
\begin{equation*}
L_n(f) = \frac{1}{nm}  \sum_{i = 1}^n \sum_{j=1}^m \rho\left( Y_{ij} - f(T_{j}) \right) + \lambda \int_0^{1} |f^{(r)}(t)|^2 dt.
\end{equation*}
Let $\NS^{2r}(T_1, \ldots, T_m)$ denote the m-dimensional natural spline subspace with knots at $T_1, \ldots, T_m$ and let $Q: \mathcal{W}^{r,2}([0,1]) \to \NS^{2r}(T_1, \ldots, T_m)$ denote the operator associated with $r$th order spline interpolation at $T_j, \ j = 1, \ldots, m$. We use $R(\mu) := \mu - Q(\mu)$ to denote the error of approximation and we write $R_j := R(\mu)(T_j), j = 1 \ldots, m$, in what follows.

Putting $g : = Q(\mu) - f$, it is easy to see that minimizing $L_n(f)$ is equivalent to minimizing  
\begin{align*}
L_n(g) = \frac{1}{nm}  \sum_{i = 1}^n \sum_{j=1}^m \rho\left( \epsilon_{ij} + R_j  + g(T_j) \right) + \lambda ||g^{(r)}||_2^2 \\ + \lambda||Q(\mu)^{(r)}||_2^2 - 2\lambda \langle Q(\mu)^{(r)}, g^{(r)} \rangle_2.
\end{align*}
Denoting $C_{n} := n^{-1} + m^{-2r} + \lambda$, we aim to show that for every $\epsilon>0$ there exists a $D_{\epsilon} \geq 1$ (in the remainder we will drop the dependence on $\epsilon$ to avoid heavy notation, i.e., we simply write $D$) such that
\begin{align}
\label{eq:ROSM9}
\lim_{n \to \infty} \Pr\left( \inf_{||g||_{r , \lambda} = D} L_n( C_{n}^{1/2} g ) > L_n(0)  \right)  \geq 1-\epsilon.
\end{align}
By Lemma \ref{Lem:ROSM1} this entails the existence of a minimizer $\widehat{g}_n$ such that $||\widehat{g}_n||^2_{r, \lambda} = O_P(C_{n})$, provided that we can establish the convexity and lower semicontinuity of $L_n$ on $\mathcal{W}^{r,2}([0,1])$, which we now check. Convexity follows easily from assumption (A1) and the convexity of the map 
\begin{align*}
f \mapsto ||f^{(r)} ||^2_2,
\end{align*}
as the composition of a convex function $f \mapsto ||f^{(r)}||_2$ and an increasing convex function on $[0, \infty)$, namely $h(x) = x^2$. To check lower semi-continuity recall that by assumption (A1) $\rho$ is convex and hence continuous, thus also lower semicontinuous. In addition, the semi-norm $||f^{(r)}||_2^2$ is continuous with respect to convergence in $\mathcal{W}^{r,2}([0,1])$. Since the sum of two continuous functions is continuous, we have thus proved the (lower semi-)continuity of the objective function. 

Using the one-to-one relation between $g$ and $f$ if \eqref{eq:ROSM9} holds we obtain
\begin{align}
\label{eq:ROSM10}
||\widehat{\mu}_n-Q(\mu)||_{r,\lambda}^2 = O_P(C_n).
\end{align}
Now, note that
\begin{align}
\label{eq:ROSM11}
||\widehat{\mu}_n - \mu||_{r, \lambda}^2 & \leq 2|| \widehat{\mu}_n - Q(\mu)||_{r, \lambda}^2 + 2|| Q(\mu) - \mu||_{r, \lambda}^2,
\end{align}
which implies that if \eqref{eq:ROSM9} holds, then the claim follows from the properties of spline interpolation. In particular, it is well-known \citep[see, e.g.,][Theorem 7.3]{Devore:1993} that 
\begin{align}
\label{eq:ROSM12}
|| \mu - Q(\mu)||_2^2 & \leq c_0 \max_{1 \leq j \leq m} |T_{j+1} - T_{j}|^{2r}  \int_{0}^1 |\mu^{(r)}(t)|^2 dt \nonumber
\\ & = O(m^{-2r}),
\end{align}
for some positive constant $c_0>0$. Here, the last bound follows from assumption (A6) and the fact that $\mu \in \mathcal{W}^{r,2}([0,1])$. On the other hand, Theorem 3.4 of \citep{Schultz:1970} implies that the spline interpolant can be chosen such that
\begin{equation}
\label{eq:ROSM13}
||(\mu-Q(\mu))^{(r)}||_2^2  \leq c_0 ||\mu^{(r)}||_2^2.
\end{equation}
for some $c_0>0$. From \eqref{eq:ROSM12} and \eqref{eq:ROSM13} we have
\begin{align}
\label{eq:ROSM14}
|| Q(\mu) - \mu||_{r, \lambda}^2  & = || Q(\mu) - \mu||_{2}^2 + \lambda || (Q(\mu) - \mu)^{(r)}||_{2}^2 \nonumber
\\ & = O(m^{-2r}) + \lambda O(1) \nonumber
\\ & = O(C_{n}),
\end{align}
since, by definition of $C_{n}$, $m^{-2r} + \lambda < C_{n}$. Thus, if \eqref{eq:ROSM9} holds, \eqref{eq:ROSM11} and~\eqref{eq:ROSM14} imply that
\begin{equation*}
||\widehat{\mu}_n - \mu||_{r, \lambda}^2 = O_P(C_{n}) + O(C_{n}) = O_P(C_{n}),
\end{equation*}
which is the desired result. 

We thus have to establish \eqref{eq:ROSM9}. To that end, decompose $L_n(C_{n}^{1/2}g) - L_n(0)$ as follows
\begin{align*}
L_n(C_{n}^{1/2}g) - L_n(0) & =   \frac{1}{nm}  \sum_{i = 1}^n \sum_{j=1}^m \rho\left( \epsilon_{ij} + R_j  +C_{n}^{1/2} g(T_j) \right) 
\\ & \quad - \frac{1}{nm}  \sum_{j=1}^m \rho\left( \epsilon_{ij} + R_j  \right) + \lambda C_{n}||g^{(r)}||_2^2  
\\ & \quad - 2\lambda C_{n}^{1/2} \langle Q(\mu)^{(r)}, g^{(r)} \rangle_2
\\ &  = \frac{1}{nm} \sum_{i=1}^m \sum_{j=1}^m \int_{R_j}^{R_j+ C_{n}^{1/2}g(T_j)} \mathbb{E}\{\psi(\epsilon_{ij} + u) \} du 
\\ & \quad \ +  \frac{1}{nm} \sum_{i=1}^m \sum_{j=1}^m \int_{R_j}^{R_j+ C_{n}^{1/2}g(T_j)} \left[\psi(\epsilon_{ij} + u)- \mathbb{E}\{\psi(\epsilon_{ij} + u) \}  \right] du
\\ & \quad \ + \lambda C_{n} ||g^{(r)}||_2^2   - 2\lambda C_{n}^{1/2} \langle Q(\mu)^{(r)}, g^{(r)}\rangle_2.
\end{align*}
After adding and subtracting $\psi(\epsilon_{ij})$ and $\mathbb{E}\{\psi(\epsilon_{ij})\} = 0$ from the second integrand, we may equivalently write
\begin{align*}
L_n(C_{n}^{1/2}g) - L_n(0) : = I_1(g) + I_2(g) + I_3(g) + I_4(g),
\end{align*}
with 
\begin{align*}
I_1(g) & : = \frac{1}{nm} \sum_{i=1}^n \sum_{j=1}^m \int_{R_j}^{R_j+ C_{n}^{1/2}g(T_j)} \mathbb{E}\{\psi(\epsilon_{ij} + u) \} du + \lambda C_{n} ||g^{(r)}||_2^2 \\
I_2(g) & : =  \frac{1}{nm} \sum_{i=1}^m \sum_{j=1}^m \int_{R_j}^{R_j+ C_{n}^{1/2} g(T_j)} \left[ \{ \psi(\epsilon_{ij} + u)  - \psi(\epsilon_{ij}) \}  \right. \\ & \left.  \quad - \mathbb{E}\{\psi(\epsilon_{ij} + u)- \psi(\epsilon_{ij}) \}   \right] du \\ 
I_3(g) & : = - C_{n}^{1/2} \frac{1}{nm} \sum_{i=1}^n \sum_{j=1}^m \psi(\epsilon_{ij}) g(T_j) \\
I_4(g) & : = -  2 C_{n}^{1/2} \lambda \langle Q(\mu)^{(r)}, g^{(r)}\rangle_2.
\end{align*}

By the superadditivity of the infimum we have the lower bound
\begin{align*}
\inf_{||g||_{r , \lambda} = D} \{L_n( C_{n}^{1/2} g ) - L_n(0)\} & \geq \inf_{||g||_{r , \lambda} = D} I_1(g) + \inf_{||g||_{r , \lambda} = D} I_2(g) 
\\ & \phantom{{}=1}+ \inf_{||g||_{r , \lambda} = D} I_3(g) + \inf_{||g||_{r , \lambda} = D} I_4(g).
\end{align*}
We will show that for $D$ sufficiently large $\inf_{||g||_{r , \lambda} = D} I_1(g)$ is positive and dominates all other terms in the decomposition. For this, we need to determine the order of each one of the four terms. 

Starting with $ \inf_{||g||_{r , \lambda} = D} I_4(g)$, since $||Q(\mu)^{(r)}||_2 \leq c_0 ||\mu^{(r)}||_2$, the Schwarz inequality yields
\begin{align}
\label{eq:ROSM15}
|\inf_{||g||_{r , \lambda} = D} I_4(g)| & \leq \sup_{||g||_{r , \lambda} = D} |I_4(g)| \nonumber
\\ & \leq 2 C_{n}^{1/2}  \lambda^{1/2} ||Q(\mu)^{(r)}||_2 \sup_{||g||_{r , \lambda} = D}\nonumber \lambda^{1/2}||g^{(r)}||_2 \\ & \leq c_0 D \lambda^{1/2} C_{n}^{1/2}  ||\mu^{(r)}||_2 \nonumber
\\ &  =  O(1) D C_{n},
\end{align}
where we have used that $\mu \in \mathcal{W}^{r, 2}([0,1])$, $\lambda^{1/2} < C_{n}^{1/2}$ and $\lambda^{1/2} ||g^{(r)}||_2 \leq ||g||_{r, \lambda}$, by definition of $||\cdot||_{r, \lambda}$.

Turning to $I_3(g)$, observe that by assumption (A2) the errors $\epsilon_{ij} = \epsilon_i(T_j)$ are independent in $i$ and identically distributed for a common $j$. With the Schwarz inequality we obtain
\begin{align*}
\left|I_3(g) \right| \leq \frac{C_n^{1/2}}{n} \left\{ \frac{1}{m} \sum_{j=1}^m \left|  \sum_{i=1}^n \psi(\epsilon_{ij}) \right|^2 \right\}^{1/2} \left\{  \frac{1}{m} \sum_{j=1}^m |g(T_j)|^2 \right\}^{1/2}.
\end{align*}
For the first factor we clearly have
\begin{align*}
\mathbb{E}\left\{\frac{1}{m} \sum_{j=1}^m \left|  \sum_{i=1}^n \psi(\epsilon_{ij}) \right|^2 \right\} & = \frac{1}{m} \sum_{j=1}^m \sum_{i=1}^n \sum_{k=1}^n \mathbb{E} \{ \psi(\epsilon_{ij}) \psi(\epsilon_{kj}) \}
\\ & = \frac{n}{m} \sum_{j=1}^m \mathbb{E} \{|\psi(\epsilon_{1j})|^2 \}
\\ & = O(n),
\end{align*}
by assumption (A5). Using Markov's inequality we thus find
\begin{align*}
\left\{ \frac{1}{m} \sum_{j=1}^m \left|  \sum_{i=1}^n \psi(\epsilon_{ij}) \right|^2 \right\}^{1/2} = O_P(n^{1/2}). 
\end{align*}
To bound the second factor, we note that as a consequence of Lemma~\ref{Lem:ROSM2}
\begin{align*}
\sup_{||g||_{r, \lambda} \leq D}  \frac{1}{m} \sum_{j=1}^m |g(T_j)|^2  & \leq \left(1 + \frac{c_r}{m \lambda^{1/2r}} \right) D^2,
\\ & \leq c_0 D^2,
\end{align*}
since, by assumption, $\liminf_{n \to \infty} m \lambda^{1/2r} >0$. Combining these two bounds we obtain
\begin{align}
\label{eq:ROSM16}
\sup_{||g||_{r, \lambda} \leq D} \left|I_3(g) \right| = C_n^{1/2} O_P(n^{-1/2}) D = O_P(1) D C_n,
\end{align}
as $n^{1/2} < C_n^{1/2}$.

Before examining the terms $I_1(g)$ and $I_2(g)$ note that, by the reproducing property and inequality \eqref{eq:ROSM8},
\begin{align*}
C_{n}^{1/2} \max_{j \leq m} |g(T_j)| = C_{n}^{1/2} \max_{j \leq m} | \langle g, \mathcal{R}_{r, \lambda}( \cdot, T_j)  \rangle_{r , \lambda}|  &\leq C_{n}^{1/2}  ||g||_{r, \lambda} \max_{j \leq m} || \mathcal{R}_{r, \lambda}( \cdot, T_j) ||_{r, \lambda} 
\\ & \leq c_0 D C_{n}^{1/2}  \lambda^{-1/4r}.
\end{align*}
Similarly, by \eqref{eq:ROSM8},
\begin{equation*}
\max_{j\leq m} |R_j| \leq c_4 \lambda^{-1/{4r}} ||R(\mu)||_{r, \lambda} \leq c_0 C_{n}^{1/2} \lambda^{-1/{4r}},
\end{equation*}
where the second  inequality follows from \eqref{eq:ROSM14}. Our limit conditions now ensure that 
\begin{equation*}
\lim_{n \to \infty} C_{n}^{1/2} \max_{j \leq m} |g(T_j)| = \lim_{n \to \infty} \max_{j\leq m} |R_j| =  0.
\end{equation*} 
We now consider $I_1(g)$. Applying Fubini's theorem and assumption (A5) yields
\begin{align*}
\int_{R_j}^{R_j+ C_{n}^{1/2}g(T_j)} \mathbb{E}\{\psi(\epsilon_{1j} + u) \} du  & =  \int_{R_j}^{R_j+ C_{n}^{1/2}g(T_j)} \{\delta_ju + o(u)\} du
\\ & =  \frac{C_{n}}{2} \delta_j |g(T_j)|^2 \{1+o(1)\}   +   \delta_jC_{n}^{1/2} R_j g(T_j)\{1+o(1) \},
\end{align*}
for all large $n$. Consequently,
\begin{align*}
\frac{1}{m}  \sum_{j=1}^m \int_{R_j}^{R_j+ C_{n}^{1/2}g(T_j)} \mathbb{E}\{\psi(\epsilon_{ij} + u) \} du & =  \frac{C_{n}}{2m} \sum_{j=1}^m \delta_j |g(T_j)|^2\{1+o(1)\}   \\ & \quad + \frac{C_{n}^{1/2}}{m} \sum_{j=1}^m \delta_j R_j g(T_j)\}\{1+o(1)\}
\\ & := I_{11}(g) + I_{12}(g),
\end{align*}
say. 

We establish a lower bound for $I_{11}(g)$ and an upper bound for $I_{12}(g)$. From assumption (A5) and Lemma \ref{Lem:ROSM2} we obtain
\begin{align}
\label{eq:ROSM17}
I_{11}(g) & \geq \inf_{j} \delta_j \frac{C_{n}}{2m} \sum_{j=1}^m |g(T_j)|^2 \{1+o(1)\} \nonumber
\\ & \geq c_0 C_{n} \left( ||g||_{2}^2 - \frac{c_r}{m \lambda^{1/2r}}||g||_{r, \lambda}^2 \right)\{1+o(1) \},
\end{align}
and, by the Schwarz inequality,
\begin{align*}
|I_{12}(g)| & \leq  C_{n}^{1/2}  \sup_{j} \delta_j \frac{1}{m} \sum_{j=1}^m |g(T_j)| |R_j| \{ 1+ o(1)\}
\\ & \leq c_0 C_{n}^{1/2} \left\{ \frac{1}{m} \sum_{j=1}^m |g(T_j)|^2 \right\}^{1/2} \left\{ \frac{1}{m} \sum_{j=1}^m |R_j|^2 \right\}^{1/2} \{1+o(1) \}.
\end{align*}
Since $||g||_{r, \lambda} \leq D$, the first factor in curly braces may be bounded by $c_0 D$ using a previous argument. For the second factor, by Lemma~\ref{Lem:ROSM2}, \eqref{eq:ROSM14} and our assumption $\liminf m \lambda^{1/2r} >0$,
\begin{align*}
\frac{1}{m} \sum_{j=1}^m |R_j|^2 & \leq ||R(\mu)||_2^2 + \frac{c_0}{m \lambda^{1/{2r}}} ||R(\mu)||_{m,\lambda}^2
\\ & = O(m^{-2r}) + O(1) O(C_n)
\\ & = O(C_n).
\end{align*}
We may thus conclude that
\begin{align}
\label{eq:ROSM18}
\sup_{||g||_{r, \lambda} \leq D}|I_{12}(g)| \leq c_0 D C_{n}\{1+o(1)\},
\end{align}
for all large $n$.

Combining \eqref{eq:ROSM17} and \eqref{eq:ROSM18} we find that, for all $g \in \{f \in \mathcal{W}^{r,2}([0,1]): ||f||_{r, \lambda} \leq D \}$,
\begin{align*}
I_1(g) &\geq  c_0 C_{n} \left( ||g||_{2}^2 - \frac{c_r}{m \lambda^{1/2r}}||g||_{r, \lambda}^2 \right)\{1+o(1) \} + C_{n} \lambda ||g^{(r)}||_2^2 + O(1) D C_{n}
\\ & = c_0 C_{n} ||g||_{r, \lambda}^2  \{1+o(1)\} + O(1) D C_{n},
\end{align*}
since, by assumption, $\liminf_{n\to \infty} m \lambda^{1/{2r}} \geq 2c_{r}$. Consequently,
\begin{align}
\label{eq:ROSM19}
\inf_{||g||_{r, \lambda} = D} I_1(g) \geq c_0 D^2 C_n \{1+O(D^{-1})+ o(1) \}.
\end{align}

To conclude the proof we now show that for every $\epsilon>0$,
\begin{align}
\label{eq:ROSM20}
\lim_{n \to \infty} \Pr \left( \sup_{||g||_{r, \lambda} \leq D} |I_2(g)| \geq \epsilon C_n \right) = 0,
\end{align}
which in combination with \eqref{eq:ROSM15}, \eqref{eq:ROSM16} and  \eqref{eq:ROSM19} establishes~\eqref{eq:ROSM9} and thus yields the result stated in the theorem. To accomplish this we introduce some notation. By assumption, the tuples $(\epsilon_{i1}, \ldots, \epsilon_{im}), i = 1, \ldots, n$, are independent and identically distributed according to $\mathbb{P}$. Moreover, the design points are common across $i$. For $g \in \{f \in \mathcal{W}^{r,2}([0,1]): ||f||_{r, \lambda} \leq D \}: = \mathcal{B}_D$ consider the class of real-valued functions $f_g: \mathbb{R}^{m} \to \mathbb{R}$ given by
\begin{align*}
f_g(\mathbf{x}) :=  \frac{1}{m} \sum_{j=1}^m \int_{R_j}^{R_j+ C_n^{1/2}g(T_j)} \{\psi(x_j + u) - \psi(x_j)\}  du.
\end{align*}
This class of functions depends on $n$ via $m$ and $C_n$ but we suppress this dependence for notational convenience. Letting $\mathbb{P}_n$ denote the empirical measure associated with $(\epsilon_{i1}, \ldots, \epsilon_{im}), i = 1, \ldots, n$, we may rewrite $I_2(g)$ as 
\begin{align*}
I_2(g) = (\mathbb{P}_n-\mathbb{P})f_g,
\end{align*}
where, as in \citep{VDV:1998}, $\mathbb{P}_n f_g$ and $\mathbb{P}f_g$ stand for expectations of $f_g$ with respect to $\mathbb{P}_n$ and $\mathbb{P}$ respectively.  By Markov's inequality we have
\begin{align}
\label{eq:ROSM21}
\Pr\left( \sup_{g \in \mathcal{B}_D} |(\mathbb{P}_n-\mathbb{P})f_g| \geq \epsilon C_n \right) \leq \frac{\mathbb{E}\{|\sup_{g \in \mathcal{B}_D}n^{1/2} |(\mathbb{P}_n-\mathbb{P})f_g|\}}{ n^{1/2} C_n \epsilon},
\end{align}
and thus it suffices to show that the right-hand side of \eqref{eq:ROSM20} tends to zero as $n \to \infty$. 

Let $N_{[ \ ]}(\epsilon, \mathcal{F}, \mathcal{L}_2(\mathbb{P}))$ denote the $\epsilon$-bracketing number for a set of functions $\mathcal{F}$ in the $\mathcal{L}^2(\mathbb{P})$-norm and let $N(\epsilon, \mathcal{F}, ||\cdot||_{\infty})$ denote the $\epsilon$-covering number in the sup-norm. Furthermore, let $J_{[ \ ]}(\delta, \mathcal{F},  \mathcal{L}_2(\mathbb{P})$ denote the $\epsilon$-bracketing integral, that is,
\begin{align*}
J_{[ \ ]}(\epsilon, \mathcal{F},  \mathcal{L}_2(\mathbb{P})) = \int_{0}^\epsilon (\log N_{[ \ ]}(u, \mathcal{F}, \mathcal{L}_2(\mathbb{P})) )^{1/2} d u.
\end{align*}
By Lemma 19.36 of \citep{VDV:1998}, for any class of real-valued functions $\mathcal{F}$ such that $\mathbb{E}\{|f|^2\} < \delta^2$ and $||f||_{\infty} \leq M$ for all $f \in 
\mathcal{F}$ we have
\begin{align}
\label{eq:ROSM22}
\mathbb{E}\{|\sup_{f \in \mathcal{F}}n^{1/2} |(\mathbb{P}_n-\mathbb{P})f| \} \leq c_0 J_{[ \ ]}(\delta, \mathcal{F},  \mathcal{L}_2(\mathbb{P})) \left(1 + \frac{J_{[ \ ]}(\delta, \mathcal{F},  \mathcal{L}_2(\mathbb{P}))}{\delta^2 n^{1/2}} M \right).
\end{align}
First we determine a bound on the bracketing number and then we estimate $\delta$ and $M$. By our previous discussion, $\max_{j \leq m} |R_j| \to 0$ and $C_n^{1/2}\sup_{g \in \mathcal{B}_D} ||g||_{\infty} \to 0$ as $n \to \infty$. Hence, for any $(g_1, g_2) \in \mathcal{B}_D \times \mathcal{B}_D$ we obtain, by assumption (A3),
\begin{align*}
\sup_{\mathbf{x} \in \mathbb{R}^m}|f_{g_1}(\mathbf{x}) - f_{g_2}(\mathbf{x})| & = \sup_{\mathbf{x} \in \mathbb{R}^m} \left|  m^{-1} \sum_{j=1}^m \int_{R_j + C_n^{1/2} g_2(T_j)}^{R_j +C_n^{1/2}g_1(T_j)}\{\psi(x_j + u) - \psi(x_j) \} du  \right| \nonumber
\\ & \leq  2 M_1 C_n^{1/2}||g_1 - g_2||_{\infty}.
\end{align*}
By Theorem 2.7.11 of \citep{VDV:1996}, we now have
\begin{align}
\label{eq:ROSM23}
N_{[ \ ]}(\epsilon, \{ f_g, g \in \mathcal{B}_D \}, \mathcal{L}_2(\mathbb{P})) \leq N(\epsilon/(4M_1 C_n^{1/2}), \mathcal{B}_D, ||\cdot||_{\infty}).
\end{align}
We develop a bound for this last covering number. By our assumptions $\lambda \in (0,1]$ for all large $n$ and therefore for all $f \in \mathcal{W}^{r,2}([0,1])$,
\begin{align*}
||f||_{r, \lambda} \geq \lambda^{1/2} ||f||_{r,1},
\end{align*}
with $||\cdot||_{r,1}$ the standard Sobolev norm. Hence,
\begin{align*}
\mathcal{B}_{D} \subset \{ f \in \mathcal{W}^{r,2}([0,1]) : ||f||_{r,1} \leq \lambda^{-1/2} D  \} 
\end{align*}
By Proposition 6 of \citep{Cucker:2001}, we deduce that for all $\epsilon>0$,
\begin{align}
\label{eq:ROSM24}
\log N(\epsilon, \mathcal{B}_{D}, ||\cdot||_{\infty}) \leq c_0 (\lambda^{1/2} \epsilon)^{-1/r}.
\end{align}
We proceed to estimate $\delta$ and $M$, as required by the Lemma. For the former observe that
by the Schwarz inequality and assumption (A4) we have
\begin{align*}
\mathbb{E}\{|f_g|^2 \} &  \leq m^{-2} \mathbb{E}\left\{ \left|  \sum_{j=1}^m \int_{R_j}^{R_j+ C_{n}^{1/2} g(T_j)} \{ \psi(\epsilon_{ij} + u)- \psi(\epsilon_{ij}) \}    du \right|^2 \right\} \nonumber
\\ &\leq  m^{-1} \sum_{j=1}^m \mathbb{E}\left\{ \left| \int_{R_j}^{R_j+ C_{n}^{1/2} g(T_j)} \{ \psi(\epsilon_{ij} + u)- \psi(\epsilon_{ij}) \}    du \right|^2 \right\} \nonumber
\\ & \leq  m^{-1} M_2 \sum_{j=1}^m C_n^{1/2}  |g(T_j)| \left|\int_{R_j}^{R_j+C_n^{1/2} g(T_j)} |u| du \right|\nonumber
\\ & \leq c_0 C_n^{1/2} \max_{j \leq m} |g(T_j)| m^{-1}  \sum_{j=1}^m \left\{ |R_j|^2 + C_n|g(T_j)|^2 \right\} \nonumber
\\ & \leq c_0  C_n^{3/2} \lambda^{-1/4r},
\end{align*}
where the last inequality follows by \eqref{eq:ROSM8}. Since this estimate is uniform on $\mathcal{B}_D$, we may take $\delta = c_0 C_n^{3/4} \lambda^{-1/8r}$. In addition, using assumption (A3) and Lemma~\ref{Lem:ROSM2} we obtain
\begin{align*}
\sup_{g \in \mathcal{B}_D} \sup_{\mathbf{x} \in \mathbb{R}^m} |f_g( \mathbf{x} )| \leq M_1 C_n^{1/2} \sup_{g \in \mathcal{B}_D} \left\{ \frac{1}{m} \sum_{j=1}^m |g(T_j)|^2 \right\}^{1/2} \leq c_0 C_n^{1/2}
\end{align*}
so that $M = c_0 C_n^{1/2}$ and consequently $M/\delta^2 \leq c_0 \lambda^{1/4r}/ C_n$. With our estimate of $\delta$ and the bound on the bracketing number implied by \eqref{eq:ROSM23}--\eqref{eq:ROSM24} the bracketing integral in \eqref{eq:ROSM22} over $n^{1/2} C_n$ may be bounded as follows
\begin{align*}
J_{[ \ ]}(\delta, \{f_g, g \in \mathcal{B}_D \},  \mathcal{L}_2(\mathbb{P}))/(n^{1/2} C_n) \leq c_0 (n^{1/2} C_n^{1/8r + 1/4} \lambda^{3/8r - 1/16r^2})^{-1}.
\end{align*}
Now,
\begin{align*}
(n^{1/2} C_n^{1/8r + 1/4} \lambda^{3/8r - 1/16r^2} )^2 \geq n \lambda^{1/2 + 1/r - 1/8r^2} \to \infty,
\end{align*}
by our limit assumptions. We have thus shown that $J_{[ \ ]}(\delta, \{f_g, g \in \mathcal{B}_D \},  \mathcal{L}_2(\mathbb{P}))/(n^{1/2} C_n) \to 0$ as $n \to \infty$. Furthermore, $\lambda \to 0$ as $n \to \infty$, which implies
\begin{align*}
\frac{M}{\delta^2 n^{1/2}} J_{[ \ ]}(\delta, \{f_g, g \in \mathcal{B}_D \},  \mathcal{L}_2(\mathbb{P})) \leq & c_0 \frac{\lambda^{1/4r}}{n^{1/2} C_n} J_{[ \ ]}(\delta, \{f_g, g \in \mathcal{B}_D \},  \mathcal{L}_2(\mathbb{P})) \\ & = o(1),
\end{align*}
as $n \to \infty$, whence we may deduce that the right-hand side of \eqref{eq:ROSM21} tends to zero. Consequently, \eqref{eq:ROSM20} holds and the proof is complete. 

\end{proof}

\begin{proof}[Proof of Corollary~\ref{Cor:ROSM1}]

The proof can be deduced from \citet[Chapter 13, Lemma 2.17]{Eg:2009} which establishes the inequality
\begin{equation*}
\{||f||^2_2 + \lambda^{s/r} ||f^{(s)}||^2_2  \}^{1/2} \leq c_{s, r} ||f||_{r , \lambda}.
\end{equation*}
for all $s \leq r$ and  $f \in \mathcal{W}^{r, 2}([0,1])$ with the constant $c_{s, r}$ depending only on $s$ and $r$. Since $\mathcal{W}^{r, 2}([0,1])$ is a vector space, Theorem \ref{Thm:ROSM1} now implies that for any $1 \leq s \leq r$
\begin{equation*}
\lambda^{s/r} ||\widehat{\mu}_n^{(s)} - \mu^{(s)}||_2^2 \leq c_{s, r} ||\widehat{\mu}_n-\mu||_{r, \lambda}^2 = O_P\left( n^{-1} + m^{-2r} + \lambda \right).
\end{equation*}
The result follows by the positivity of $\lambda$.	

\end{proof}

To prove Theorem~\ref{Thm:ROSM2} we use the following uniform law of large numbers on Sobolev spaces.

\begin{lemma}
\label{Lem:ROSM3}
If assumption (B6) holds, $\lambda \to 0$ and $n \lambda^{(3-1/2r)/(2r)} \to \infty$ as $n \to \infty$, then for every fixed $D>0$ and $\delta>0$ the following uniform law holds
\begin{align*}
\lim_{n \to \infty} \Pr \left( \sup_{||g||_{r, \lambda} \leq D} \left| \frac{1}{n} \sum_{i=1}^n \frac{1}{m_i} \sum_{j=1}^{m_i} |g(T_{ij})|^2 - \int_{0}^1 |g(t)|^2 dt \right| \geq \delta \right) = 0
\end{align*}

\begin{proof}

First observe that the set $\mathcal{B}_D :=\{g \in \mathcal{W}^{r,2}([0,1]): ||g||_{r, \lambda} \leq D \}$ is compact in the $||\cdot||_{\infty}$-topology and from the proof of Theorem~\ref{Thm:ROSM1} its entropy satisfies
\begin{align*}
\log (N(\delta, \mathcal{B}_{D}, ||\cdot||_{\infty})) \leq c_0 \lambda^{-1/2r} \delta^{-1/r}.
\end{align*} 
Furthermore, by assumption (B6) $\mathbb{E}\{|g(T_{ij})|^2\} = ||g||_2^2$ and by \eqref{eq:ROSM8}, $\sup_{g \in \mathcal{B}_D}||g||_{\infty}^2 \leq  c_0 \lambda^{-1/2r}$, where $c_0$ depends on $D$. For $i=1, \ldots, n$, define independent stochastic processes
\begin{align*}
U_{i,g^2} = m_i^{-1} \sum_{j=1}^{m_i} |g(T_{ij})|^2 - \int_0^1 |g(t)|^2 dt,\quad g \in \mathcal{B}_{D},
\end{align*}
so that
\begin{align*}
\frac{1}{n} \sum_{i=1}^n \frac{1}{m_i} \sum_{j=1}^{m_i} |g(T_{ij})|^2 - \int_{0}^1 |g(t)|^2 dt = \frac{1}{n}\sum_{i=1}^n U_{i,g^2}.
\end{align*}
Fix $g_0 \in \mathcal{B}_D$ and observe that
\begin{align}
\label{eq:ROSM25}
\sup_{g \in \mathcal{B}_D}  \left| \frac{1}{n}\sum_{i=1}^n U_{i,g^2} \right|  \leq  \left|\frac{1}{n}\sum_{i=1}^n U_{i,g_0^2} \right| + \sup_{g \in \mathcal{B}_{D}} \left| \frac{1}{n}\sum_{i=1}^n U_{i,g^2} - \frac{1}{n} \sum_{i=1}^n U_{i,g_0^2} \right|.
\end{align}
We first deal with the first term on the right hand side of \eqref{eq:ROSM25} and show that it converges to zero in probability. Note that each $U_{i,g_0^2}$ is uniformly bounded for $g_0 \in \mathcal{B}_{D}$ as
\begin{align*}
\max_{1 \leq i \leq n} |U_{i,g_0^2}| \leq c_0 \lambda^{-1/2r},
\end{align*}
for some $c_0$ not depending on $n$. For every $\delta>0$, Hoeffding's inequality \citep[see, e.g.,][]{van de Geer:2000} now yields
\begin{align*}
\Pr\left(\left|\frac{1}{n} \sum_{i=1}^n U_{i,g_0^2} \right| \geq \delta \right) \leq 2 \exp\left[ - c  \delta^2 n \lambda^{1/r} \right].
\end{align*}
Our limit assumptions imply $n \lambda^{1/r} \to \infty$ as $n \to \infty$ and therefore the exponential tends to zero. To treat the second term in \eqref{eq:ROSM24} we first note that
\begin{align*}
\max_{1 \leq i \leq n}|U_{i,g^2}-U_{i, \tilde{g}^2}| \leq 2 ||g^2 - \tilde{g}^2||_{\infty}, \quad g, \tilde{g} \in \mathcal{B}_{D},
\end{align*}
and there exists a $c_0>0$ such that
\begin{align*}
||g^2 - \tilde{g}^2||_{\infty} \leq c_0 \lambda^{-1/4r}||g-\tilde{g}||_{\infty}, \quad g, \tilde{g} \in \mathcal{B}_{D}.
\end{align*}
This shows that
\begin{align*}
N(\delta, \{g^2, g\in  \mathcal{B}_{D}\}, ||\cdot||_{\infty})) \leq N(  c_0 \delta \lambda^{1/4r} , g\in  \mathcal{B}_{D}\}, ||\cdot||_{\infty})).
\end{align*}
From our bound on this last covering number it now follows that 
\begin{align*}
\int_{0}^{c_0 \lambda^{-1/2r}} ( \log (N(u, \{g^2, g \in \mathcal{B}_{D}\}, ||\cdot||_{\infty}) )^{1/2} du & \leq c_0 \lambda^{-1/8r^2} \lambda^{-1/4r} \int_{0}^{c_0 \lambda^{-1/2r}} u^{-1/2r} du
\\ & \leq c_0 \lambda^{-1/8r^2} \lambda^{-1/4r} \lambda^{-1/2r} \lambda^{1/4r^2}
\\ & = c_0 \lambda^{-3/4r} \lambda^{1/8r^2}.
\end{align*}
By our limit assumptions, $n \lambda^{3/2r - 1/4r^2} \to \infty$. Therefore, applying Lemma 8.5 of \citet{van de Geer:2000} (with $d_i(g^2, g_0^2) = ||g^2 -g_0^2||_{\infty}$) yields
\begin{align*}
\Pr\left( \sup_{g \in \mathcal{B}_{D}} \left| \frac{1}{n}\sum_{i=1}^n U_{i,g^2} - \frac{1}{n} \sum_{i=1}^n U_{i,g_0^2} \right| \geq \delta \right) \leq c_0 \exp[- c_0n \delta^2 \lambda^{1/r}],
\end{align*}
for some $c_0>0$ which completes the proof.

\end{proof}

\end{lemma}
\begin{proof}[Proof of Theorem~\ref{Thm:ROSM2}]

Let $L_{n}(f)$ denote the objective function, that is,
\begin{equation*}
L_{n}(f) := \frac{1}{n} \sum_{i = 1}^n \frac{1}{m_i} \sum_{j=1}^{m_i} \rho\left( Y_{ij} - f(T_{ij}) \right) + \lambda \int_0^{1} |f^{(r)}(t)|^2 dt,
\end{equation*}
and write $C_{n} = (n m \lambda^{1/2r})^{-1} + 1/n +  \lambda$. It will be shown that for every $\epsilon>0$ there exists a $D_{\epsilon} \geq 1$ such that
\begin{equation}
\label{eq:ROSM26}
\lim_{n\to \infty} \Pr\left(\inf_{||g||_{r,\lambda} = D} L_n(\mu + C_{n}^{1/2}g) > L_n(\mu) \right) \geq 1-\epsilon.
\end{equation}
As before, this then implies the existence of a minimizer in the ball $\{f \in \mathcal{W}^{r,2}([0,1]): ||f-\mu||_{r,\lambda}^2 \leq D C_{n} \}$ with high probability. 

Decomposing $L_n(\mu + C_{n}^{1/2}g) - L_n(\mu)$ yields
\begin{align*}
L_n(\mu + C_{n}^{1/2}g) - L_n(\mu) & = \frac{1}{n} \sum_{i=1}^n \frac{1}{m_i} \sum_{j=1}^{m_i} \int_{0}^{-C_{n}^{1/2} g(T_{ij})} \{\psi(\epsilon_{ij} + u) - \psi(\epsilon_{ij}) \} du
\\ & \quad  - \frac{C_{n}^{1/2}}{n} \sum_{i=1}^{n} \frac{1}{m_i} \sum_{j=1}^{m_i}  g(T_{ij}) \psi(\epsilon_{ij})  + \lambda C_{n}   ||g^{(r)}||_2^2
\\ & \quad + 2 \lambda  C_{n}^{1/2}  \langle \mu^{(r)}, g^{(r)} \rangle_{2}
\\ & := I_1(g) + I_2(g) + I_3(g),
\end{align*}
with
\begin{align*}
I_1(g) & : = \frac{1}{n} \sum_{i=1}^n \frac{1}{m_i} \sum_{j=1}^{m_i} \int_{0}^{-C_{n}^{1/2} g(T_{ij})} \{\psi(\epsilon_{ij} + u) - \psi(\epsilon_{ij}) \} du + C_{n} \lambda ||g^{(r)}||_2^2
\\ I_2(g) & : =  \frac{C_{n}^{1/2}}{n} \sum_{i=1}^{n} \frac{1}{m_i} \sum_{j=1}^{m_i}  g(T_{ij}) \psi(\epsilon_{ij})
\\ I_3(g) & := 2 \lambda  C_{n}^{1/2} \langle \mu^{(r)}, g^{(r)} \rangle_{2}.
\end{align*}
Letting $\mathcal{T} = \{T_{ij} \}_{i,j}$ we proceed to show the following:
\begin{align}
\inf_{||g||_{r, \lambda} = D}\mathbb{E}\{ I_1(g) |  \mathcal{T} \} & \geq c_0 D^2 C_n \{1+o_P(1)\} \label{eq:ROSM27}
\\  \sup_{||g||_{r, \lambda} \leq D} |I_1(g)-\mathbb{E}\{ I_1(g) |  \mathcal{T} \}| &= o_P(1) C_n \label{eq:ROSM28}
\\ \sup_{||g||_{r, \lambda} \leq D} |I_2(g)| & = O_P(1) D C_n \label{eq:ROSM29}
\\ \sup_{||g||_{r, \lambda} \leq D} |I_3(g)| & = O(1) D C_n. \label{eq:ROSM30}
\end{align}
These bounds are sufficient for \eqref{eq:ROSM26} to hold, as by virtue of \eqref{eq:ROSM27} $\inf_{||g||_{r, \lambda} = D}\mathbb{E}\{ I_1(g) |  \mathcal{T} \}$ will be positive and dominate all other terms for sufficiently large $D$. The bound on \eqref{eq:ROSM30} can be readily verified with the argument employed in the proof of Theorem~\ref{Thm:ROSM1}, so it remains to establish \eqref{eq:ROSM27}-- \eqref{eq:ROSM29}.

Starting with \eqref{eq:ROSM29}, note that with the reproducing property we may write
\begin{align*}
|I_2(g)| & = C_n^{1/2} \left| \frac{1}{n} \sum_{i=1}^n \frac{1}{m_i} \sum_{j=1}^{m_i} \psi(\epsilon_{ij}) \langle g, \mathcal{R}_{r, \lambda}(T_{ij}, \cdot) \rangle_{r, \lambda}  \right| \\ &=  C_n^{1/2}\left|  \Big \langle  g, \frac{1}{n} \sum_{i=1}^n \frac{1}{m_i} \sum_{j=1}^{m_i} \psi(\epsilon_{ij}) \mathcal{R}_{r, \lambda}(T_{ij}, \cdot) \Big \rangle_{r, \lambda}  \right|.
\end{align*}
Consequently, by the Schwarz inequality,
\begin{align*}
|I_2(g)| \leq C_n^{1/2} \left\| \frac{1}{n} \sum_{i=1}^n \frac{1}{m_i} \sum_{j=1}^{m_i} \psi(\epsilon_{ij}) \mathcal{R}_{r, \lambda}(T_{ij}, \cdot)  \right\|_{r, \lambda} ||g||_{r, \lambda}.
\end{align*}
Now, by assumptions (B3), (B5) and (B6),
\begin{align*}
\mathbb{E}\left\{  \left\| \frac{1}{n} \sum_{i=1}^n \frac{1}{m_i} \sum_{j=1}^{m_i} \psi(\epsilon_{ij}) \mathcal{R}_{r, \lambda}(T_{ij}, \cdot)  \right\|_{r, \lambda}^2 \bigg| \mathcal{T} \right\}  \nonumber
& =  \frac{1}{n^2}\sum_{i=1}^n \sum_{j=1}^{m_i} \sum_{l=1}^{m_i} \frac{1}{m_i^2} \mathcal{R}_{r, \lambda}(T_{ij}, T_{il}) \mathbb{E}\{\psi(\epsilon_{ij}) \psi(\epsilon_{il}) | \mathcal{T} \}
\\ & = \frac{1}{n^2} \sum_{i=1}^n \sum_{j=1}^{m_i} \frac{1}{m_i^2 } \mathcal{R}_{r, \lambda}(T_{ij}, T_{ij}) \mathbb{E}\{|\psi(\epsilon_{1}(T_j))|^2  | \mathcal{T}  \} \nonumber
\\ & \quad + \frac{1}{n^2} \sum_{i=1}^n \sum_{j \neq l} \frac{1}{m_i^2} \mathcal{R}_{r, \lambda}(T_{ij}, T_{il}) \mathbb{E}\{\psi(\epsilon_{1}(T_{ij})) \psi(\epsilon_{1}(T_{il})) | \mathcal{T} \}
\\ & := Z_1 + Z_2,
\end{align*}
say. For $Z_1$ we immediately find
\begin{align*}
Z_1 \leq \frac{c_0}{n^2 \lambda^{1/2r}} \sum_{i=1}^{n} \frac{1}{m_i} = O \left( \frac{1}{n m \lambda^{1/2r}} \right),
\end{align*}
by assumption (B5) and \eqref{eq:ROSM8}. Hence,
\begin{align*}
\mathbb{E}\left\{  \left\| \frac{1}{n} \sum_{i=1}^n \frac{1}{m_i} \sum_{j=1}^{m_i} \psi(\epsilon_{ij}) \mathcal{R}_{r, \lambda}(T_{ij}, \cdot)  \right\|_{r, \lambda}^2  \right\} = O \left( \frac{1}{n m \lambda^{1/2r}} \right) + \mathbb{E}\{Z_2\}.
\end{align*}
To bound $ \mathbb{E}\{Z_2\}$ note that, by assumptions (B3) and (B6),
\begin{align*}
\mathbb{E}\{Z_2\}  = \frac{1}{n^2} \sum_{i=1}^n \frac{m_i (m_i-1)}{m_i^2} \int_{0}^1 \int_{0}^1 \mathcal{R}_{r, \lambda}(x,y) \mathbb{E}\{\psi(\epsilon_1(x)) \psi(\epsilon_1(y)) \} dx dy.
\end{align*}
Now, let $T_{\mathcal{R}_{r, \lambda}}$ denote the integral operator that maps $\mathcal{L}^2([0,1])$ into itself with $\mathcal{R}_{r, \lambda}$ as its kernel, viz,
\begin{align*}
(T_{\mathcal{R}_{r, \lambda}}f)(x) = \int_{0}^1 \mathcal{R}_{r, \lambda}(x,y) f(y) dy, \quad f \in \mathcal{L}^2([0,1]).
\end{align*}
By Proposition~\ref{Prop:ROSM1}, $T_{\mathcal{R}_{r, \lambda}}$ is positive semidefinite and compact with largest eigenvalue equal to $1$.  Combining these observations with assumption (B5) we find
\begin{align*}
\mathbb{E}\{Z_2\} \leq  n^{-1} \mathbb{E}\{ \langle T_{\mathcal{R}_{r, \lambda}}\psi(\epsilon_1), \psi( \epsilon_1) \rangle_2 \} \leq n^{-1} \mathbb{E}\{ ||\psi(\epsilon_1)||_2^2 \} = O(n^{-1}),
\end{align*}
see, e.g., \citep[Theorem 4.2.5]{Hsing:2015} for the second inequality. Applying Markov's inequality now yields
\begin{align*}
\sup_{||g||_{r, \lambda} \leq D} |I_2(g)| = O_P(1) D C_n,
\end{align*}
as asserted.

Turning to $I_1(g)$, a similar derivation using the reproducing property as in the proof of Theorem \ref{Thm:ROSM1} shows that
\begin{equation}
\label{eq:ROSM31}
\lim_{n \to \infty } C_{n}^{1/2}  \max_{1 \leq i \leq n} \max_{1 \leq j \leq m_i} |g(T_{ij})| = 0,
\end{equation}
under our limit conditions. Applying Fubini's theorem and assumption (B5) yields
\begin{align*}
\mathbb{E} \left\{ \int_{0}^{-C_{n}^{1/2} g(T_{ij})} \{\psi(\epsilon_{ij} + u) - \psi(\epsilon_{ij}) \} du \bigg| \mathcal{T}  \right\}  & =  \int_{0}^{-C_{n}^{1/2} g(T_{ij})}  \{\delta_{ij}u + o_P(u)\} du  
\\ & = 2^{-1}\delta_{ij}  C_{n} |g(T_{ij})|^2\{1+o_P(1) \}.
\end{align*}
Consequently,
\begin{align}
\label{eq:ROSM32}
\mathbb{E}\{I_1(g) | \mathcal{T} \} \geq \inf_{i,j}  \delta_{ij} 2^{-1}  C_n \frac{1}{n} \sum_{i=1}^n \frac{1}{m_i} \sum_{j=1}^{m_i} |g(T_{ij})|^2 \{1 + o_P(1) \} + C_n ||g^{(r)}||_2^2,
\end{align}
for all sufficiently large $n$ with arbitrarily high probability. Our limit conditions imply those of Lemma~\ref{Lem:ROSM3} and therefore,
\begin{align*}
\frac{1}{n} \sum_{i=1}^n \frac{1}{m_i} \sum_{j=1}^{m_i} |g(T_{ij})|^2 & = \frac{1}{n} \sum_{i=1}^n \frac{1}{m_i} \sum_{j=1}^{m_i} \mathbb{E} |g(T_{ij})|^2 \\ & \quad + \left(\frac{1}{n} \sum_{i=1}^n \frac{1}{m_i} \sum_{j=1}^{m_i} |g(T_{ij})|^2 - \frac{1}{n} \sum_{i=1}^n \frac{1}{m_i} \sum_{j=1}^{m_i} \mathbb{E} |g(T_{ij})|^2  \right)
\\ & = \int_{0}^1 |g(t)|^2 dt + o_P(1)
\end{align*}
where the $o_P(1)$ term is uniform in $g \in \mathcal{B}_D$. From \eqref{eq:ROSM32} it follows that
\begin{align*}
\inf_{||g||_{r, \lambda} = D}\mathbb{E}\{ I_1(g) |  \mathcal{T} \} & \geq c_0 D^2 C_n \{1+o_P(1)\},
\end{align*}
so we have established \eqref{eq:ROSM27}. 

To complete the proof we need to show \eqref{eq:ROSM28}, for which we use an empirical process result given by Theorem 8.13 of \citep{van de Geer:2000}. Let $(\Omega, \mathcal{A}, \mathbb{P})$ denote the underlying probability space and write $\mathbf{T}_{i} = (T_{i1}, \ldots, T_{im_i})^{\top}$ and $\boldsymbol{\epsilon}_i = (\epsilon_{i1}, \ldots, \epsilon_{im_i})^{\top}$. Let $\mathcal{F}_i = \sigma((\mathbf{T}_1, \boldsymbol{\epsilon}_1), \ldots, (\mathbf{T}_i, \boldsymbol{\epsilon}_i)), i = 1, \ldots$ be an increasing sequence of sigma fields generated by the $\mathbf{T}_i$ and the $\boldsymbol{\epsilon}_i$. Clearly $\mathcal{F}_i \subset \mathcal{A}$ for all $i$, as the $\mathbf{T}_i$ and $\boldsymbol{\epsilon}_i$ are measurable. We have
\begin{align*}
I_1(g) - \mathbb{E}\{I_1(g) | \mathcal{T}\} = n^{-1} \sum_{i=1}^n  Z_{i,g},
\end{align*}
with $Z_{i,g}$ independent mean-zero random variables given by
\begin{align*}
Z_{i,g} & = m_i^{-1} \sum_{j=1}^{m_i} \int_{0}^{-C_n^{1/2}g(T_{ij})} \left[\{\psi(\epsilon_{ij} + u) - \psi(\epsilon_{ij})\}  - \mathbb{E}\{ \psi(\epsilon_{ij}+u) - \psi(\epsilon_{ij}) | \mathcal{T} \}  \right]du.
\end{align*}
The $Z_{i,g}$ are $\mathcal{F}_i$-measurable since, by assumption (B6),
\begin{align*}
\mathbb{E}\{ \psi(\epsilon_{ij}+u) - \psi(\epsilon_{ij}) | \mathcal{T} \} = \mathbb{E}\{ \psi(\epsilon_{ij}+u) - \psi(\epsilon_{ij}) | \mathbf{T}_i \}.
\end{align*}
Furthermore, the $Z_{i,g}$ are uniformly bounded since $C_n^{1/2} \sup_{g \in \mathcal{B}_{D}}||g||_{\infty} = o(1)$ and consequently, by assumption (B3),
\begin{align}
\label{eq:ROSM33}
\sup_{g \in \mathcal{B}_{D}}\max_{1 \leq i \leq n} |Z_{i,g}| \leq  c_0 C_n^{1/2} \lambda^{-1/4r}.
\end{align}
for all large $n$. Thus, the required expectations exist and
\begin{align*}
S_n := \sum_{i=1}^n Z_{i,g}, \quad \mathcal{F}_1 \subset \mathcal{F}_2 \subset \ldots
\end{align*}
is a martingale. Let $\mathbf{Z}_{g} = (Z_{1,g}, \ldots, Z_{n,g})^{\top}$ and write $\rho_K$ for the Bernstein "norm", i.e.,
\begin{align*}
\rho_K(Z_{i,g}) = 2K^2 \mathbb{E}\{e^{|Z_{i,g}|/K} -1 - |Z_{i,g}|/K \}, \quad i = 1, \ldots, n,
\end{align*}
for some $K>0$. Furthermore, define the "average" squared Bernstein norm
\begin{align*}
\bar{\rho}_K^2(\mathbf{Z}_g) = n^{-1} \sum_{i=1}^n \rho_K(Z_{i,g}).
\end{align*}

It is helpful to recall the definition of the generalized entropy with bracketing \citep{van de Geer:2000}. For $0<\delta \leq R$ and $\mathbf{F} \in \mathcal{A}$, this is defined as the logarithm of the smallest non-random value of $N$ for which a collection of pairs of random vectors $\mathbf{Z}_{k}^{L} = (Z_{1,k}^{L}, \ldots, Z_{n,k}^{L})^{\top}$ and 
$\mathbf{Z}_{k}^{U} = (Z_{1,k}^{U}, \ldots, Z_{n,k}^{U})^{\top}$ with $[Z_{i,k}^L, Z_{i,k}^U] $ $\mathcal{F}_i$-measurable, $i = 1 \ldots, n, k = 1, \ldots, N,$ such that for all $g \in \mathcal{B}_{D}$, there is a $j = j(g) \in \{1, \ldots, N\}$ with $g \mapsto j(g)$ non-random, such that
\begin{itemize}
\item[(i)] $\bar{\rho}_K^2(\mathbf{Z}_j^{U} - \mathbf{Z}_j^{L}) \leq \delta^2$ on $\{\bar{\rho}_K(\mathbf{Z}_g) \leq R \} \cap \mathbf{F}$.
\item[(ii)] $Z_{i,k}^L \leq Z_{i,g} \leq Z_{i,k}^U, i =1\ldots, n,$ on $\{\bar{\rho}_K(\mathbf{Z}_g) \leq R \} \cap \mathbf{F}$.
\end{itemize}  
We denote this value of $N$ with $\mathcal{N}_{B,K}(\delta, R, \mathbf{F})$ and its logarithm, the generalized entropy, with $\mathcal{H}_{B,K}(\delta, R, \mathbf{F})$. 

We take $\mathbf{F} = \Omega$ for our purposes and develop a bound on the generalized entropy with bracketing. In particular, for specific $K$ and $R$ we show that
\begin{align}
\label{eq:ROSM34}
\mathcal{H}_{B,K}(\delta, R, \Omega) \leq \log N(c_0 \delta/C_n^{1/2}, \mathcal{B}_{D}, ||\cdot||_{\infty}),
\end{align}
where $c_0$ is a generic constant. To prove \eqref{eq:ROSM34} we take $K = c_0 C_n^{1/2} \lambda^{-1/4r}$ where $c_0$ is the constant given in \eqref{eq:ROSM32} and we determine a bound on $\bar{\rho}_K(\mathbf{Z}_{g})$, that is, we determine a value of $R$. Since, by \eqref{eq:ROSM33}, $|Z_{i,g}| \leq K$, Lemma 5.8 of \citep{van de Geer:2000} implies
\begin{align*}
\rho_{K}(Z_{i,g}) \leq c_0 \sup_{g \in \mathcal{B}_{D}} \max_{1 \leq i \leq n} \mathbb{E}\{ |Z_{i,g}|^2 \}.
\end{align*}
Now, applying the Schwarz inequality twice along with assumption (B4) yields
\begin{align*}
\mathbb{E}\{|Z_{i,g}|^2 |\mathcal{T} \} & \leq m_i^{-2} \mathbb{E}\left\{ \left| \sum_{j=1}^{m_i}  \int_{0}^{-C_{n}^{1/2} g(T_{ij})} \{\psi(\epsilon_{ij} + u) - \psi(\epsilon_{ij}) \} du \bigg| \mathcal{T} \right|^2 \right\} \nonumber
\\ & \leq  m_{i}^{-1} \sum_{j=1}^{m_i} \mathbb{E} \left\{ \left|  \int_{0}^{-C_{n}^{1/2} g(T_{ij})} \{\psi(\epsilon_{ij} + u) - \psi(\epsilon_{ij}) \} du \bigg| \mathcal{T} \right|^2 \right\}
 \\ & \leq c_0 m_{i}^{-1} \sum_{j=1}^{m_i} C_n^{1/2}|g(T_{ij})|\left|\int_{0}^{-C_{n}^{1/2} g(T_{ij})} |u|  du \right| \nonumber
\\ & \leq c_0  \lambda^{-1/4r} C_n^{3/2} m_{i}^{-1} \sum_{j=1}^{m_i} |g(T_{ij})|^2.
\end{align*}
Iterating expectations and using assumption (B6), we obtain
\begin{align*}
\sup_{g \in \mathcal{B}_D} \max_{1 \leq i \leq n}\mathbb{E}\{|Z_{i,g}|^2 \} \leq c_0 \lambda^{-1/4r} C_n^{3/2}.
\end{align*}
Thus, we may take $R = c_0 \lambda^{-1/8r} C_n^{3/4}$ in the definition of generalized entropy with bracketing. The above also reveals that $\{ \bar{\rho}_{K}(\mathbf{Z}_{g}) \leq R \} = \Omega$. For this $R$, let us now establish \eqref{eq:ROSM34}. 

Fix $\delta>0$ and let $g_1, \ldots, g_N$ denote a $\delta$-net covering $\mathcal{B}_{D}$ in the supremum norm. Noting that for every $g_1, g_2 \in \mathcal{B}_{D}$ we have
\begin{align*}
\max_{1 \leq i \leq n}|Z_{i,g_1} - Z_{i,g_2}| & \leq c_0 C_n^{1/2}||g_1-g_2||_{\infty}.
\end{align*}
It is easy to see that for each $g \in \mathcal{B}_{D}$ there exists a non-random $g_k$ with $k \in \{1, \ldots, N\}$ such that 
\begin{align*}
Z_{i,g_k} - \delta c_0C_n^{1/2} \leq Z_{i,g} \leq Z_{i,g_k} + \delta c_0C_n^{1/2}, \quad 1, \ldots, n.
\end{align*}
Take $Z_{i,k}^L = Z_{i,g_k} - \delta c_0 C_n^{1/2}$ and $Z_{i,k}^U = Z_{i,g_k} + \delta c_0 C_n^{1/2}$. Since $\lambda^{-1/4r} \geq 1$ for all large $n$, applying Lemma 5.8 of \citet{van de Geer:2000} leads to 
\begin{align*}
\rho_{K}(Z_{i,k}^U - Z_{i,k}^L) = \rho_{K}(2 \delta c_0 C_n^{1/2}) \leq c_0 \delta^2 C_n.
\end{align*}
The constant $c_0$ does not depend on $i$ and consequently after averaging we still have
\begin{align*}
\bar{\rho}_K^2(\mathbf{Z}_{k}^U -\mathbf{Z}_{k}^L ) \leq c_0 \delta^2 C_n,
\end{align*}
where $\mathbf{Z}_{k}^U = (Z_{1,k}^U, \ldots, Z_{n,k}^U)^{\top}$ and likewise $\mathbf{Z}_{k}^L = (Z_{1,k}^L, \ldots, Z_{n,k}^L)^{\top}$. We have thus shown that a cover of $\mathcal{B}_{D}$ with radii at most $\delta/(c_0 C_n^{1/2})$ provides a set of $\delta$-brackets according to the definition of generalized entropy. In other words,
\begin{align*}
\mathcal{H}_{B,K}(\delta, R, \Omega) \leq \log N(c_0 \delta/C_n^{1/2}, \mathcal{B}_{D}, ||\cdot||_{\infty}),
\end{align*}
for some $c_0$, as was to be shown. 

With these values of $K$, $R$ and the bound on $\mathcal{H}_{B,K}(\delta, R, \Omega)$ developed above, we now check the conditions of Theorem 8.13 in \citet{van de Geer:2000}. First,
\begin{align*}
n^{1/2} C_n = o(n^{1/2}R^2/K),
\end{align*}
as $n \to \infty$, or, equivalently,
\begin{align*}
C_n^{1/2} = o(1),
\end{align*}
as $n \to \infty$, which holds, by virtue of $\lambda \to 0$ and $n \lambda^{1/2r} \to \infty$. Furthermore, by \eqref{eq:ROSM34}, the bracketing integral may be bounded as follows
\begin{align*}
\int_{0}^{c_0 C_n^{3/4} \lambda^{-1/8r}} (\mathcal{H}_{B,K}(u, c_0 C_n^{3/4} \lambda^{-1/8r}, \Omega))^{1/2} du & \leq  \int_{0}^{c_0 C_n^{3/4} \lambda^{-1/8r}} (\log N(c_0 u/C_n^{1/2},\mathcal{B}_{D}, ||\cdot||_{\infty})^{1/2} du
\\ & = c_0 \frac{C_n^{1/4r}}{\lambda^{1/4r}} \int_{0}^{c_0 C_n^{3/4} \lambda^{-1/8r}} u^{-1/2r} du.
\end{align*}
Evaluating the latter integral we find
\begin{align*}
\int_{0}^{c_0 C_n^{3/4} \lambda^{-1/8r}} (\mathcal{H}_{B,K}(u, c_0 C_n^{3/4} \lambda^{-1/8r}, \Omega))^{1/2} du & \leq c_0 \frac{C_n^{1/4r}}{\lambda^{1/4r}} C_n^{3/4} \lambda^{-1/8r} C_n^{-3/8r} \lambda^{1/16r^2}
\\ & = c_0 C_n^{3/4-1/8r} \lambda^{-3/8r + 1/16r^2}.
\end{align*}
Consequently,
\begin{align*}
\int_{0}^{c_0 C_n^{3/4} \lambda^{-1/8r}} (\mathcal{H}_{B,K}(u, c_0 C_n^{3/4} \lambda^{-1/8r}, \Omega))^{1/2} du / (n^{1/2} C_n) \leq c_0 n^{-1/2} C_n^{-1/4 - 1/8r} \lambda^{-3/8r + 1/16r^2}.
\end{align*}
The definition of $C_n$ now yields
\begin{align*}
(n^{1/2} C_n^{1/4+1/8r} \lambda^{3/8r-1/16r^2})^2 \geq n \lambda^{1/2+1/r-1/8r^2} \to \infty,
\end{align*}
by our limit assumptions. The conditions of the theorem are thus satisfied, which implies that
\begin{align*}
\Pr\left( n^{1/2} \sup_{g \in \mathcal{B}_{D}}|I_1(g) - \mathbb{E}\{I_1(g) | \mathcal{T}\}| \geq \epsilon n^{1/2} C_n \right) \leq c_0 \exp\left[ - c_0 \epsilon^2 n C_n^{1/2} \lambda^{1/4r}  \right],
\end{align*}
for every $\epsilon>0$. The exponential tends to zero, as our limit assumptions imply $ n C_n^{1/2} \lambda^{1/4r}  \to \infty$ as $n \to \infty$. Hence \eqref{eq:ROSM28} holds and the result follows. 

\end{proof}

\end{document}